\documentclass[acmsmall, screen, nonacm]{acmart}

\newbool{extendedversion}
\setbool{extendedversion}{true}
\newbool{reviewversion}
\setbool{reviewversion}{false}

\setcopyright{cc}
\setcctype{by-nc-nd}
\acmDOI{10.1145/3808306}
\acmYear{2026}
\acmJournal{PACMPL}
\acmVolume{10}
\acmNumber{PLDI}
\acmArticle{228}
\acmMonth{6}
\acmSubmissionID{pldi26main-p402-p}
\received{2025-11-13}
\received[accepted]{2026-04-03}

\usepackage[nameinlink]{cleveref}
\usepackage{tikz}
\usetikzlibrary{positioning,tikzmark}
\usepackage{annotate-equations}
\usepackage{multirow}
\usepackage{subcaption}
\usepackage{wrapfig}
\usepackage{listings}
\usepackage{lst-ocaml}
\usepackage{refable-rules}
\usepackage{soteria-macros}
\usepackage{proof-macros}

\crefname{section}{}{}
\crefname{subsection}{}{}
\creflabelformat{section}{#2Sec.~#1#3}
\creflabelformat{subsection}{#2\S#1#3}

\begin{document}

\title{Soteria: Efficient Symbolic Execution as a Functional Library\ifbool{extendedversion}{\ifbool{reviewversion}{}{ (Extended Version)}}{}}
\subtitle{Perhaps you \emph{should} write your own symbolic execution engine!}

\author{Sacha-Élie Ayoun}
\orcid{0000-0001-9419-5387}
\affiliation{%
  \institution{Imperial College London}
  \city{London}
  \country{United Kingdom}
}
\affiliation{%
  \institution{Soteria Tools Ltd.}
  \city{London}
  \country{United Kingdom}
}
\email{s.ayoun17@imperial.ac.uk}

\author{Opale Sjöstedt}
\orcid{0009-0003-7545-1383}
\affiliation{%
  \institution{Imperial College London}
  \city{London}
  \country{United Kingdom}
}
\affiliation{%
  \institution{Soteria Tools Ltd.}
  \city{London}
  \country{United Kingdom}
}
\email{opale.sjostedt23@imperial.ac.uk}

\author{Azalea Raad}
\orcid{0000-0002-2319-3242}
\affiliation{%
  \institution{Imperial College London}
  \city{London}
  \country{United Kingdom}
}
\affiliation{%
  \institution{Soteria Tools Ltd.}
  \city{London}
  \country{United Kingdom}
}
\email{azalea.raad@imperial.ac.uk}


\ifbool{extendedversion}{}{
}

\begin{abstract}
Symbolic execution (SE) tools often rely on intermediate languages (ILs) to support multiple programming languages, promising reusability and efficiency. 
In practice, this approach introduces trade-offs between performance, accuracy, and language feature support. 
We argue that building SE engines \emph{directly} for each source language is both simpler and more effective. 
We present \sot, a lightweight \ocaml library for writing SE engines in a functional style, without compromising on performance, accuracy or feature support. 
\sot enables developers to construct SE engines that operate directly over source-language semantics, offering \emph{configurability}, compositional reasoning, and ease of implementation. 
Using \sot, we develop \rusteria, the \emph{first} Rust SE engine supporting \TBs (the intricate aliasing model of Rust), and \coteria, a compositional SE engine for C. 
Both tools are competitive with or outperform state-of-the-art tools such as Kani, \pulse, CBMC and Gillian-C in performance and the number of bugs detected. 
We formalise the theoretical foundations of \sot and prove its soundness, demonstrating that sound, efficient, accurate, and expressive SE can be achieved without the compromises of ILs. 
\end{abstract}

\ifbool{extendedversion}{}{
}




\maketitle


\section{Introduction}

In the last two decades, a number of symbolic execution (SE) tools were designed with the intent of supporting \emph{multiple languages}, including \infertool~\cite{infer-automatic-program-calcagno-2011,finding-real-bugs-le-2022} for Java, C, Hack, and Python, 
Gillian~\cite{gillian-part-multilanguage-fragososantos-2020,gillian-part-ii-maksimovic-2021} for JavaScript, C, and Rust~\cite{hybrid-approach-semiautomated-ayoun-2025}, as well as others~\cite{viper-verification-infrastructure-muller-2016,klee-unassisted-automatic-cadar-2008,owi-performant-parallel-andres-2024}.
To do this, each such tool relies on an \emph{intermediate language} (IL) and offers an SE engine
that can execute IL code. 
A putative advantage of having an IL is that one can avoid re-implementing an SE engine for each new language L: instead, one can obtain it by \emph{compiling} L to IL. 
In this work, we argue that one \emph{should} implement an SE engine for each language and start  
by addressing two arguments commonly made in favour of ILs.

\paragraph{Argument 1: Design once, use often!}
This oft-touted mantra of ILs refers to their reusability when adding support for a new language. 
In practice, however, designing both an efficient and sufficiently expressive IL is far from straightforward, 
and it requires \emph{foresight} to predict all possible features in future languages. 
In our experience, this then results in either a bloated IL that is \emph{inefficient}, or an \emph{incomplete} IL that cannot support the desired features. 
In the latter case,  one can either revisit the IL to extend it with the required features, defeating the goal of `design once'; forgo supporting these features altogether, resulting in incomplete reasoning; or extend the compiler with complex pre-analyses to offset the lack of support in the IL, resulting in \emph{ad hoc solutions}.

Practically all three of these IL-related issues---inefficiency, incompleteness, and ad hoc solutions---are illustrated in the various attempts of existing tools to 
analyse Rust using ILs, such as CBMC's GOTO~\cite{cbmc-bounded-model-kroening-2014} (used by Kani~\cite{kani-verifying-dynamic-trait-vanhattum-2022})
or GIL (used by Gillian-Rust~\cite{hybrid-approach-semiautomated-ayoun-2025}). 
For example, GOTO does not support detecting uninitialised memory accesses, which are undefined behaviours (UBs), and so the Kani compiler performs an entire points-to analysis to enable this detection, which still remains incomplete. Further, detecting aliasing violations as per the intricate \TBs~\cite{tree-borrows-villani-2025} aliasing model of Rust remains entirely beyond the reach of Kani. 
On the other hand, Gillian-Rust encodes Rust pointers and mutable borrows in creative ways because GIL does not have the constructs to express these concepts naturally, hindering the expressivity and efficiency of the analysis.

Note that although using an \emph{existing} IL (\eg LLVM-IR or Wasm) eschews the need for a dedicated compiler, it often leads to \emph{inaccurate} analyses. 
This is because existing compilers are not \emph{semantics-preserving}~\cite{c11-optimasations-are-invalid} and erase UBs and other information deemed irrelevant to the IL. 
As such, analysing the resulting IL code may lead to \emph{false negatives}. 
For instance, Wasm does not encode any information required for reasoning about \TBs, and hence Owi~\cite{owi-performant-parallel-andres-2024}, an SE engine with Wasm as its IL, cannot find aliasing bugs. 
Moreover, compilation causes Owi's IL code to grow substantially in size, which is likely to make its Rust analysis less efficient. 
All of this leads us to our first claim:

\smallskip
\noindent\textbf{Claim 1}: \emph{IL-design is a trade-off between performance, accuracy, and comprehensive support for language features, and one cannot maximise all aspects of this trio at once. A dedicated analysis can.}

\paragraph{Argument 2: Prove once, use often!} From the meta-theoretical point of view, an IL-based SE engine can be proven sound against its IL once and for all, and then all existing and future analyses built over the IL should inherit its soundness guarantees. 
This, however, makes a \emph{significant assumption}, namely that the source-to-IL compiler is also sound. 
Not only is this often not the case in practice~\cite{miri-practical-undefined-behavior-jung-2026}, verifying (proving sound) a compiler designed for analysis is a highly challenging task~\cite{trustworthy-automated-program-parthasarathy-2024}.
Moreover, as discussed above, compilers often erase UBs and other information unsupported or deemed irrelevant by the IL, leading to gaps in expressivity or even \emph{false negatives} in the analysis (\eg such false negatives occur in Gillian-C, see \cref{sec:soteria-c}). 
This brings us to our second claim:

\smallskip
\noindent\textbf{Claim 2}:
\emph{A trustworthy interpreter for \textnormal{L} is not harder to write than a trustworthy \textnormal{L}-to-\textnormal{IL} compiler.}

\paragraph{The \sot Framework}
Motivated by our two claims above, we present \emph{\sot}, a simple, yet powerful, \ocaml library for writing SE engines. 
A key novelty of \sot is that it abstracts the fixed-IL aspect of the analysis, allowing one to construct an SE engine that operates \emph{directly} on the source language or on an IL close to the source-language semantics. 
\sot provides a \emph{monadic interface}, together with intuitive syntactic sugar, to write SE engines in functional style. 
To show the utility of \sot, we build two SE engines over it: \rusteria, our automated symbolic testing engine for Rust; and \coteria, our automated, compositional, bi-abduction engine for C (see below). 
Using these two case studies, we show the advantages of \sot along several dimensions.

First, we show that \sot yields \emph{highly performant} engines comparable or better than the state of the art. 
For example, \rusteria competes in performance with Kani (a Rust SE engine maintained by a team of AWS engineers over the last four years)~\cite{kani-verifying-dynamic-trait-vanhattum-2022}, while covering a significantly larger subset of the Rust semantics (see \cref{sec:soteria-rust}).
To our knowledge, \rusteria is the \emph{first} SE engine to support detecting bugs related to the \TBs aliasing model~\cite{tree-borrows-villani-2025}, while \emph{no existing verification tool}~\cite{leveraging-rust-types-astrauskas-2019,creusot-foundry-deductive-denis-2022,refinedrust-type-system-gaher-,verus-verifying-rust-lattuada-2023,hybrid-approach-semiautomated-ayoun-2025,flux-liquid-types-lehmann-2022,aeneas-rust-verification-ho-2022,modular-formal-verification-foroushaani-2022} supports \TBs.
Notably, using \rusteria we detected and reported a potential bug in \myinline{hashbrown}~\cite{hashbrown-2025}, the \emph{second-most-downloaded} Rust crate~\cite{crates-io-stats} (see \cref{sec:soteria-rust}).

Second, we show that building SE engines over \sot is simple: as presented here, \rusteria took a first-year PhD student \emph{merely eight months} to implement.
This is partly due to the simple monadic interface of \sot that allows one to write SE engines in a functional style that closely resembles a concrete interpreter for the source language, and partly
because \sot provides a library of \emph{reusable components} that can be used across engines.
Indeed, the student re-used many of the components already developed for \coteria in \rusteria.

Finally, we formalise the soundness guarantees of \sot and the engines built over it. 
Specifically, following previous work~\cite{gillian-foundations-implementation-ayoun-2024}, we \emph{formalise} the theoretical foundation of \sot and prove that the symbolic interpreters built over it are sound as long as they meet certain conditions (\cref{sec:formalisation}).

We further highlight that \sot is \emph{highly flexible and configurable} in three ways. 
First is the choice of the \emph{source language}: clients of \sot can build analyses over it for any language L, so long as they develop an interpreter for L. 
More notably, they inherit the \emph{soundness guarantees} of \sot by construction, provided they meet the conditions of \cref{sec:formalisation}. 

Second, \sot is designed to come with \emph{batteries included} and provides out-of-the-box support for logging and statistics.
Indeed, logs in \sot are often more informative than those in IL-based frameworks, as they relate directly to the relevant source code rather than to the IL.

Third, inspired by \citewithauthor{compositional-symbolic-execution-loow-2024}, \sot can be \emph{configured} to perform analysis in one of two \emph{modes}: \emph{under-approximate} (\mux) analysis for \emph{bug catching}~\cite{incorrectness-logic-ohearn-2019,local-reasoning-presence-raad-2020} (as in Pulse~\cite{finding-real-bugs-le-2022}), or \emph{over-approximate} (\mox) analysis for \emph{bounded verification} (as in \eg CBMC).
Specifically, the analysis mode $m$ is passed to the \emph{SE monad} at the core of \sot, which determines the execution path exploration strategy \ala $m$ (exploring \emph{all} paths in \mox and \emph{some} paths up to a depth in \mux). 
Importantly, our monadic approach is also compatible with \emph{unbounded} verification using compositional SE~\cite{verified-symbolic-execution-keuchel-2022,gillian-foundations-implementation-ayoun-2024}.

Lastly, \sot is flexible enough to implement various analyses as long as they can be expressed as symbolic execution. Bounded whole-program symbolic testing (WPST) is one such analysis (as in CBMC~\cite{cbmc-bounded-model-kroening-2014}); bug-finding by means of under-approximate \emph{bi-abduction} (as in \pulse~\cite{finding-real-bugs-le-2022}) is another (as per~\cite{javert-20-compositional-fragososantos-2019,compositional-symbolic-execution-loow-2024,compositional-symbolic-execution-next-loow-2025}).
For instance, \coteria supports \emph{both} WPST and fully automatic bi-abductive bug detection. 
As we show in \cref{sec:soteria-c}, WPST in \coteria is competitive in performance with CBMC.
Moreover,  its compositional capabilities are competitive in performance with the industry-grade \infertool.\pulse (henceforth \pulse) tool, while being \emph{an order of magnitude faster} than Gillian-C~\cite{gillian-part-ii-maksimovic-2021}. 

\paragraph{Contributions and Outline}
Our contributions, described intuitively in \cref{sec:overview}, are as follows.
\begin{enumerate}
	\item[(\cref{sec:implem})] We present \sot, a functional library for building \emph{automated, industry-scale} SE engines. 
	\item[(\cref{sec:soteria-rust})] We describe \rusteria, our Rust SE engine that outperforms the state of the art in bug finding capabilities,  and is the \emph{first} SE engine to support \TBs. 
	\item[(\cref{sec:soteria-c})] We present \coteria, our SE engine for C that supports \emph{both} WPST \emph{and} fully automated bug detection using bi-abduction.  \coteria is comparable to the industry-grade \pulse tool and outperforms Gillian-C (factor of 2) and CBMC (by an order of magnitude).
	\item[(\cref{sec:formalisation})] We formalise the theory behind \sot and prove the soundness of its guarantees.
	\item[(\cref{sec:discussions})] We discuss key design choices we made in \sot and their limitations. 
	\item[(\cref{sec:relwork})] We conclude with a discussion of related work.
\end{enumerate}

\section{Overview: Build Your Own Symbolic Execution Engine for \simlang over \sot}\label{sec:overview}
\begin{figure}[t]
\centering
\includegraphics[width=0.8\textwidth]{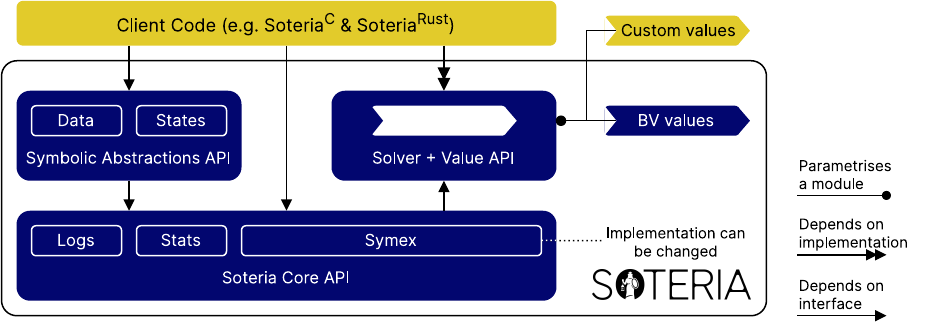}
\vspace{-10pt}
\caption{Overview of the \sot architecture}
\vspace{-10pt}
\label{fig:architecture}
\end{figure}

We depict an overview of the \sot architecture in \cref{fig:architecture}.
\sot is an \emph{efficient}, functional programming library written in \ocaml that provides:
\begin{itemize}
\item a monadic interface with all required primitives to write a symbolic execution (SE) engine;
\item an efficient implementation of this interface, parametrised on a user-defined value language;
\item a pre-defined instantiation of the value language for convenience;
\item pragmatic, developer-friendly tooling for aggregating statistics and logs; and
\item ready-made data structures for SE that can be soundly used within symbolic computations.
\end{itemize}

We present a didactic overview of \sot as a tutorial on using it to implement a \emph{symbolic interpreter}. We do this for a simple side-effect-free expression language, \simlang.
Our aim is not to give an exhaustive account of \sot's features, but rather an intuitive account of using it.
For lack of space we elide the details unnecessary for understanding the intuition behind \sot.
Nevertheless, we present the full details in the \ifbool{extendedversion}{appendix (\cref{app:simple-lang})}{extended version~\cite{extended-version}}.
The \simlang syntax\ifbool{extendedversion}{ (in \cref{app:simple-lang})}{} is standard and comprises the expected constructs of let-binding, if-else branching, and so forth.

\paragraph{Preliminiary: Custom Let-Operators in \ocaml} \sot makes extensive use of \ocaml's \emph{custom let-operators}, which (though daunting at first glance) are simply syntactic sugar for a monadic \myinline{bind} operator. Specifically,  given a monad \myinline{M}, one can define \myinline{let ( let* ) = M.bind}, whereafter \ocaml will desugar each subsequent occurrence of \myinline{let* x = e1 in e2} to \myinline{M.bind e1 (fun x -> e2)}.

\begin{wrapfigure}{r}{0.47\textwidth}
\vspace{-1em}
\begin{minipage}{0.47\textwidth}
\begin{lstlisting}
let ( let* ) = ExecutionMonad.bind
let eval subst expr =
match expr with
| If (cond, then_e, else_e) ->
  let* cond_v = eval subst cond in
  let* cond_b = bool_of_val cond_v in
  if cond_b then eval subst then_e
            else eval subst else_e
\end{lstlisting}%
\end{minipage}
\vspace{-1.5em}
\end{wrapfigure}
\paragraph{Preliminary: Concrete Monadic Interpreter}
In the snippet across we define a simple \emph{concrete monadic interpreter} for the \myinline{If} construct of \simlang.
We define the interpreter using \ocaml syntax and introduce concepts required to later understand \sot.
First, we assign the monadic bind operator of the \myinline{ExecutionMonad} (defined in \ifbool{extendedversion}{appendix \cref{app:simple-lang}}{the extended version~\cite{extended-version}}) to the \myinline{let*} operator.
Effectively, \myinline{let* x = e1 in e2} executes the monadic computation \myinline{e1}, which may be non-deterministic and may error (\ie it returns a set of results that are either values or errors).
Each error result of \myinline{e1} is propagated to the set of final results;
for each successful (non-error) result $v$, it continues by executing \myinline{e2} with \myinline{x} bound~to~$v$.

When interpreting the \myinline{If} construct, we first evaluate the guard \myinline{cond}; once again it may error and yield multiple values.
For each successful result \myinline{cond_v}, we convert it to an \ocaml boolean \myinline{cond_b} using \myinline{bool_of_val}, which may also error if \myinline{cond_v} is not a boolean (\ie if the evaluated program is ill-typed).
Finally, if \myinline{cond_b} is true, we evaluate expression \myinline{then_e}; otherwise, we evaluate  \myinline{else_e}.

Note that this interpreter, while written in \ocaml \emph{syntax}, might not be computable: evaluating an expression may result in infinitely many results, \eg evaluating \texttt{NondetInt} can yield \emph{any} integer.
For instance,  the program below yields the (infinite) set of results
$\{(\ook, x) \mid x  >5 \} \cup \{(\oerr, x) \mid x \leq 5\}$.

\begin{lstlisting}
Let x = NondetInt in (If (x > 5) then Ok x else Error x)
\end{lstlisting}

\begin{wrapfigure}{r}{0.48\textwidth}
\begin{minipage}{0.48\textwidth}
\vspace{-1em}
\begin{lstlisting}[escapeinside={&}{&}]
&\tikzmark{sym-let-start}&let ( let* ) = Symex.Result.bind&\tikzmark{sym-let-end}&
let eval subst expr =
match expr with
| If (cond, then_e, else_e) ->
  let* cond_v = eval subst cond in
  let* cond_b = bool_of_val cond_v in
  if&\tikzmark{sat-start}&%sat&\tikzmark{sat-end}& cond_b then eval subst then_e
                else eval subst else_e
\end{lstlisting}%
\end{minipage}
\newcommand{\thisfigureletxoffset}{-.3em}
\newcommand{\thisfigureletyoffset}{-48.6em}
\newcommand{\thisfiguresatxoffset}{-.3em}
\newcommand{\thisfiguresatyoffset}{-50.6em}
\tikz[overlay,remember picture]{
  \draw[fill=softgreen!80,draw=none,opacity=.3]
    ($(pic cs:sym-let-start) + (-.2em,.9em) + (\thisfigureletxoffset, \thisfigureletyoffset)$) rectangle ($(pic cs:sym-let-end) + (.2em,-.4em) + (\thisfigureletxoffset, \thisfigureletyoffset)$);
}%
\tikz[overlay,remember picture]{
  \draw[fill=softgreen!80,draw=none,opacity=.3]
    ($(pic cs:sat-start) + (0em,.9em) + (\thisfiguresatxoffset, \thisfiguresatyoffset)$) rectangle ($(pic cs:sat-end) + (0em,-.4em) + (\thisfiguresatxoffset, \thisfiguresatyoffset)$);
}%
\vspace{-1.9em}
\end{wrapfigure}
\paragraph{Symbolic Interpreter with \sot}
Lifting our concrete interpreter above to a \emph{symbolic} one is straightforward, as shown across.
It is almost identical to the concrete interpreter, except that \begin{enumerate*}
\item \myinline{let*} is now bound to the \myinline{Symex.Result.bind} operator,  sequencing \emph{symbolic computations}; and
\item instead of branching on a concrete boolean using \myinline{if}, we branch on a \emph{symbolic} one using \myinline{if\%sat}.
\end{enumerate*}

A symbolic computation is a computation that depends on \emph{symbolic variables}, which represent unknown values.
A symbolic computation yields a set of results; each result is associated with a \emph{path condition}, a symbolic boolean (boolean expression with symbolic variables) constraining the symbolic variables under which the result occurs.
For instance, consider the example above.
During SE, \myinline{NondetInt} will yield a \emph{single branch} where \myinline{x} is bound to a fresh symbolic variable $\symb{x}$ representing an unknown integer, with the path condition $\vtrue$ (\ie with no constraint on $\symb{x}$).
Indeed, this is a key advantage of SE: infinitely many branches of concrete execution can be represented by a finite number of symbolic branches with symbolic variables.
The \myinline{Symex.bind} operator sequences symbolic computations by simply threading the path conditions appropriately within each branch.

The \myinline{if\%sat} construct (which is syntactic sugar for the \sot \myinline{Symex.branch_on} primitive -- see \cref{sec:implem}), lifts the concrete notion of conditional branching to the symbolic world.
Its condition is a \emph{symbolic} boolean (rather than a concrete one), \ie a boolean expression that may depend on symbolic variables.
It then behaves as expected for SE.
First it checks if the guard ($\symb{x} > 5$ in the example above) is satisfiable; if so, then it executes the \myinline{then} branch.
Next, it also checks if the negation of the guard ($\symb{x} \leq 5$ in the example above) is satisfiable; if so, then executes the \myinline{else} branch.
As such, the SE of the example above returns two branches:
\begin{enumerate*}
	\item $\braket{\ook: \symb{x}}{\symb{x} > 5}$, the (then) branch resulting in the successful result $\ook: \symb{x}$ with path condition $\symb{x} > 5$;
	\item $\braket{\oerr: \symb{x}}{\symb{x} \leq 5}$, the (else) branch resulting in the error result $\oerr: \symb{x}$ with path condition $\symb{x} \leq 5$.
\end{enumerate*}

\paragraph{The \sot Symex Module}
The symbolic interpreter described above is implemented using our \myinline{Symex} module in \sot.
In addition to the \myinline{let*} and \myinline{if\%sat}, the \myinline{Symex} module provides a number of primitives that enable users to write \emph{efficient} symbolic interpreters from scratch with \emph{minimal effort}.
These primitives include \eg \myinline{nondet} for introducing fresh symbolic variables, or \myinline{branches} for creating multiple branches from a list of computations. The latter is useful for modelling real non-determinism (\eg allocation that may succeed or fail). Finally, \myinline{Symex} provides the \myinline{run} function that receives three arguments:
\begin{enumerate*}
	\item a symbolic computation;
	\item an `execution mode' that determines whether the analysis is \emph{over-approximate} (OX) or \emph{under-approximate} (UX) (inspired by~\cite{compositional-symbolic-execution-loow-2024}); and
	\item an optional \emph{fuel} to limit the breadth and depth of execution (infinite by default). It then executes the computation and returns a list of results, each paired with its path condition.
\end{enumerate*}

\sot provides pre-defined common data structures (\myinline{Data} module) and state monad transformers (\myinline{State} module) that can be soundly used within symbolic computations, enabling code sharing between various interpreters.
Through this reusability and its \myinline{Symex} module,  \sot simplifies the task of developing \emph{efficient} SE engines.
For instance, as we show in \cref{sec:soteria-rust}, this enabled a \emph{first-year} PhD student to develop a highly efficient Rust SE engine that is competitive in performance to the state-of-the-art tool Kani~\cite{kani-verifying-dynamic-trait-vanhattum-2022} by AWS (developed and maintained by \emph{seasoned industry engineers} in the last \emph{four years}),  while covering a substantially larger subset of Rust, including \TBs~\cite{tree-borrows-villani-2025}.

\paragraph{Symbolic Values and Solvers}
\sot is \emph{highly general} in that users can obtain a \emph{bespoke \myinline{Symex} module tailored} to their needs.
Specifically, \sot provides a \emph{functor}, \myinline{Soteria.Symex.Make}: given a module implementing the \sot \myinline{Solver} interface, it produces a \myinline{Symex} module tailored to the symbolic value language that the solver can reason about.
In the case of \simlang above, symbolic values need only to range over integers and booleans;
as such, the symbolic value language supplied to \myinline{Soteria.Symex.Make} only needs to implement these two sorts, as well as operations that may cause branching, \eg integer comparison. Given such a solver, let us call it \myinline{Simple_solver}, the \myinline{Symex} module above can be simply obtained by a \emph{single line}: \myinline{module Symex = Soteria.Symex.Make (Simple_solver)}.

In practice, users of \sot may often want to re-use an existing solver implementation rather than implementing their own from scratch.
To this end, \sot comes with a pre-defined solver module \myinline{BV_solver} (described later in \cref{par:bv_values}) that implements a symbolic value language based on bit-vectors,  comes with several important \emph{optimisations}, and uses Z3 as its underlying SMT solver.
That is, while \sot affords its users the \emph{flexibility} to supply their own solvers, it also comes packaged with \emph{efficient, general purpose and ready-to-use} components out of the box.

\paragraph{Guarantees}
The \sot \myinline{Symex} monad provides important soundness guarantees, summarised as follows.
\begin{enumerate*}
  \item Primitives of the \myinline{Symex} monad are sound against their concrete counterparts; \eg \myinline{Symex.nondet} indeed abstracts over concrete non-determinism;
  \item sequential composition of sound symbolic computations is sound against the sequential composition of their concrete counterparts;
  \item the symbolic \myinline{if\%sat} construct is sound against the concrete \myinline{if} construct.
\end{enumerate*}
While the user must ensure they compose sound primitives that faithfully capture the concrete semantics of their interpreted language, we believe that doing so is no more difficult than writing a \emph{correct} compiler.

Finally,  \sot provides out-of-the-box support for logging and statistics.
Logging in \sot is directly integrated in SE: users can log messages at any point of execution, and \sot can create an HTML report that groups logs by execution path (see \ifbool{extendedversion}{\cref{app:logging-impl}}{extended version~\cite{extended-version}} for an example).

\section{Implementation}\label{sec:implem}


Recall from \cref{sec:overview} that \sot comprises several components as depicted in \cref{fig:architecture}. 
In what follows, we detail the implementation of the \myinline{Solver}, \myinline{Symex}, and \myinline{Data} modules.  
We highlight several \emph{optimisations} we have implemented in \sot, and we signpost them clearly with labels \displayoptimcounter{i}.
These optimisations, while minor in isolation, have been carefully designed to work in synergy and have a significant impact on the overall performance of symbolic execution.

\subsection{Solver and Values}

Recall that symbolic execution (SE) uses symbolic variables to abstract over sets of values.
As discussed in \cref{sec:overview}, \sot is parametric in its choice of symbolic variables.
Specifically,  instantiating symbolic variables in \sot
requires implementing the \myinline{Value} and \myinline{Solver} modules. 
The \myinline{Value} module describes the carrying type of symbolic values (or symbolic expressions), while the \myinline{Solver} module describes a stateful object to which constraints (\ie boolean symbolic values) can be added, and which has a \myinline{check_sat} function checking whether the current set of constraints is satisfiable.

These two modules have clearly defined interfaces, and the arrows in \cref{fig:architecture} indicate that \myinline{Symex} depends on the \emph{interfaces} of \myinline{Solver} and \myinline{Value}, while clients of \sot may additionally depend on the \emph{implementation} of \myinline{Value}, which is necessary to construct symbolic values in the first place. 

\paragraph{The Value Interface} The value interface (slightly simplified here for clarity) exposes a minimal set of operations and types required to manipulate symbolic values during SE.
It exposes two types: 
\begin{enumerate*}
	\item \myinline{t}, carrying values (a popular \ocaml convention is to name the main type of a module as `\myinline{t}'), and
	\item \myinline{ty}, carrying physical representations of the type of symbolic values.
\end{enumerate*}
The interface also exposes \myinline{mk_var : Var_id.t -> ty -> t} to create symbolic variables from a variable identifier and a type; \myinline{not : t -> t} to negate symbolic booleans, necessary for defining \myinline{if\%sat} where both the guard and its negation must be checked for satisfiability; and \myinline{as_bool : t -> bool option} to concretise a symbolic value as a concrete boolean where possible (and otherwise returning \myinline{None}). 
The interface also contains operations for pretty-printing and variable substitution (omitted for brevity). 

%
%
\paragraph{The Solver Interface} The solver interface defines a `solver' of type \myinline{t} as an object to which constraints (boolean symbolic values) can be added using \myinline{add_constraints}. 
One can check whether the current state satisfies these constraints using \myinline{check_sat}, which returns $\sat$, $\unsat$ or $\unknown$ (the \myinline{Solver_result.t} type). 
The current state also records the symbolic variables currently in scope, allowing to create \emph{fresh} variables using \myinline{fresh_var}. 
A solver must be \emph{incremental}: it must support saving the current state using \myinline{save}, backtracking the last \myinline{n} checkpoints using \myinline{backtrack_n}, and resetting to the initial state using \myinline{reset}. 
The solver must provide a way to copy the current state of the solver as a list of boolean symbolic values using \myinline{as_values}. 
Finally, it must provide the \myinline{simplify} procedure to simplify a given symbolic value according to the current set of constraints; this is particularly important to enable optimisations in the SE monad, as we discuss below.

\begin{lstlisting}
module type Solver = sig
  module Value : Value.S
  type t
  val add_constraints : t -> Value.(sbool t) list -> unit
  val check_sat : t -> Solver_result.t
  val fresh_var : t -> 'a Value.ty -> Var_id.t
  val save : t -> unit
  val backtrack_n : t -> int -> unit
  val reset : t -> unit
  val as_values : t -> 'a Value.t list
  val simplify : t -> 'a Value.t -> 'a Value.t
end
\end{lstlisting}

Note that this interface is already implemented by SMT solvers such as Z3 or CVC5, which support incremental solving via push/pop commands. As such, one can simply wrap such a solver directly behind this interface and use the solver expression language as the symbolic value representation. 
However, by keeping the \myinline{Value} opaque, we allow for arbitrary abstract implementations of symbolic values, which may enable simplifications that are more tailored to the SE use case. 
Similarly, using the \myinline{Solver} interface directly with an SMT solver is inefficient in practice, as SMT solvers are often designed to solve large batches of constraints at once, rather than incrementally adding small constraints. These two insights motivated the design of the \myinline{BV_solver} implementation.

\paragraph{BV\_values}\label{par:bv_values} 
The \myinline{Bv_Value} module \emph{implements} the \myinline{Value} interface, using bitvectors to represent most symbolic values. 
Specifically, a symbolic value's type is either a bitvector of a given size, an IEEE float of a given precision, a location (memory object identifier) represented as a bitvector, a pointer represented as a location-offset pair (both bitvectors) à la CompCert~\cite{compcert-memory-model-leroy-2012}, or a boolean. 
This is sufficient to represent all values required to symbolically execute C and Rust programs.

Our first optimisation (\newoptim{opt:hashcons}) is to hashcons~\cite{type-safe-modular-hash-consing-filliatre-2006} all values to reduce memory consumption, speed up syntactic equality checks, and enable efficient maps and sets of symbolic values. 
Our second optimisation (\newoptim{opt:reduce-on-constr}) consists in hiding the data representation of symbolic values to force clients to use constructors that eagerly perform simplifications, following work from \citewithauthor{qed-computing-what-correnson-2014}. 
For instance, constructing the addition of two concrete bitvectors immediately returns their sum as a concrete bitvector. 
Note that \optimref{opt:hashcons} and \optimref{opt:reduce-on-constr} synergise: simplifications ensure that certain values cannot be constructed in the first place, normalising them to a canonical form, hence increasing memory efficiency of the hashconsing table by reducing the number of distinct nodes.

\paragraph{BV\_solver} The \myinline{BV_solver} module implements the \myinline{Solver} interface using a pipeline composed of lightweight analyses inspired by abstract interpretation, manually optimised state management, and an underlying connection to the Z3 SMT solver.

\begin{wrapfigure}{r}{0.37\textwidth}
\vspace{-1.5em}
\begin{minipage}{0.37\textwidth}
\begin{lstlisting}
type t = {
  var_counter: Var_counter.t;
  analyses: Analysis.t;
  state: Solver_state.t;
  z3_conn: Z3_conn.t }
\end{lstlisting}
\end{minipage}
\vspace{-2em}
\end{wrapfigure}

The solver uses a variable counter to track which symbolic variables are currently in scope and to create fresh ones when requested.
In our next optimisation (\newoptim{opt:analysis}), when constraints are added using \myinline{add_constraints}, they are first passed to \myinline{Analysis.t}, which comprises a few \emph{lightweight analyses} (inspired by abstract domains) to store a subset of the current set of constraints \emph{efficiently}. 
For instance, our \emph{interval analysis} tracks upper and lower bounds on bitvector variables, allowing us to detect unsatisfiable constraints (\eg \myinline{0 < x && x < 0}) quickly without calling the SMT solver, or to simplify constraints such as \myinline{0 <= x <= 0} to \myinline{x = 0}.
Our \emph{equality analysis} tracks equalities between expressions and uses a cost function to then decide which expression to use as a canonical representative when simplifying constraints.

After going through the analyses, constraints that have not been fully stored are added to the \myinline{Solver_state.t}, which is a vector of constraints. 
Saving and backtracking the solver state is implemented by simply recording the size of this vector at each checkpoint.

\begin{wrapfigure}{r}{0.39\textwidth}
\vspace{-1em}
\begin{minipage}{0.39\textwidth}
\begin{lstlisting}
let process =
  let* () = Symex.assume c1 in
  let* () = Symex.assume c2 in
  if%sat c3 then ...
\end{lstlisting}
\end{minipage}
\vspace{-1.3em}
\end{wrapfigure}
Further, in our next optimisation (\newoptim{opt:check-mark}), each constraint in the solver state is marked as `checked' or `unchecked'. 
When \myinline{check_sat} is called, all unchecked constraints, together with all constraints that share variables with them, are sent to Z3 for checking; if Z3 returns $\sat$, then all constraints are marked as checked. 
This avoids sending the \emph{entire} set of constraints to Z3 at each call to \myinline{check_sat}, which would be very inefficient in practice. 
\optimref{opt:check-mark} allows us to add several constraints to the path condition using \myinline{assume} \emph{without} incurring the cost of a check when they are added.
To see this, consider the following code snippet on the right.

The two calls to \myinline{assume} simply add \myinline{c1} and \myinline{c2} to the solver state as unchecked constraints. When reaching the \myinline{if\%sat}, a checkpoint is created,  \myinline{c3} is added to the solver state and all three constraints are sent to the solver to check for satisfiability.
If the result is satisfiable, all three constraints are marked as checked. 
When backtracking to the checkpoint, \myinline{c3} is removed and \myinline{not c3} is added to the solver state. 
Since \myinline{c1} and \myinline{c2} are already marked as checked, only \myinline{not c3} and constraints sharing variables with it need to be sent to the solver, which may save us from checking \myinline{c1} and \myinline{c2} again.

Additionally, as part of our next optimisation (\newoptim{opt:cache}), we \emph{cache} all queries previously sent to Z3; 
in our Collections-C case study (\cref{sec:soteria-c}), this reduces Z3 calls by a factor of 3.8 and execution time by a factor of 3. 
This further synergises with \optimref{opt:hashcons}, since implementing the cache becomes cheap. It also synergises well with \optimref{opt:reduce-on-constr} as simplifications imply a higher probability of cache hits. 
In fact, we also add construction-time transformations of symbolic values that aim to optimise caching more than solving time, \eg by always choosing the same ordering of operands in commutative operations.

Finally, our last optimisation (\newoptim{opt:contextual-reductions}) uses a \myinline{simplify} function which leverages analyses, as well as the solver state, to perform \emph{contextual} simplifications. 
For instance, a constraint that is already part of the solver state is simplified to \myinline{true}, and its negation is simplified to \myinline{false}.
Surprisingly, such simple optimisations allow for significant reductions in calls to the SMT solver in practice.



We have also developed a lightweight implementation of the Solver interface without \optimref{opt:hashcons}--\optimref{opt:contextual-reductions}.
Our lightweight implementation is directly wrapped around Z3, using its push/pop functionality to manage the path condition. 
However, we found this lightweight approach to be orders of magnitude slower than \myinline{BV_solver} empirically.  
We leave a more in-depth comparison as future work.

\subsection{The Symex Module}
The \myinline{Symex} module provides the core primitives required to define SE. 
We write a clear interface for \myinline{Symex} so that different implementations of the SE monad can be swapped easily. 
For instance, one could design monads that explore branches in different orders (\eg depth- versus breadth-first exploration). 
We present the \sot \myinline{Symex} interface and our implementation of it. 

\paragraph{The Symex Interface}
The \myinline{Symex} interface exposes only core primitives required for efficient SE.
Specifically,  it exposes 
\begin{enumerate*}
	\item \myinline{return} and \myinline{bind}, the monadic operations for sequencing SE steps; 
	\item \myinline{branch_on}, the desugared function underpinning \myinline{if\%sat}; 
	\item \myinline{branches : 'a t list -> 'a t}, to explore a list of branches without conditions, required to model non-determinism, \eg allocation, which may return either a valid allocated object or a null pointer, and \myinline{val vanish: unit -> 'a Symex.t}, which stops exploration of the current branch;
	\item \myinline{val nondet : Value.ty -> Value.t t}, to instantiate an unconstrained symbolic value for a given type; and
	\item \myinline{run}, to run an SE monad to completion, returning a list of pairs, each comprising a final result and the path condition (as a list of constraints) leading to this result. 
\end{enumerate*}
As shown below, the \myinline{run} function also takes an optional \emph{fuel} parameter, infinite by default, to limit the number of branches or steps explored during SE, as well as a \emph{mode} parameter to choose between over- (\mox) and under-approximate (\mux) SE.

\begin{lstlisting}
val run:?fuel:Fuel.t -> mode:Approx.t -> 'a t -> ('a * Value.(sbool t)) list
\end{lstlisting}

While these are the only core primitives \emph{required} to define any symbolic computation, the \myinline{Symex} interface exposes additional operations either for convenience, or because they allow for more efficient implementations. 
These include \myinline{assume}, which adds a constraint to the path condition, and \myinline{assert}, which checks if a constraint holds and stops exploring the current branch if it does not. 
These two operations can be implemented using \myinline{if\%sat} as follows:
\begin{lstlisting}
let assume cond = if%sat cond then return () else vanish ()
let assert cond = if%sat cond then ok () else error AssertError
\end{lstlisting}
However, providing them as primitive operations allows for more efficient implementations that avoid exploring unnecessary branches. 
For instance, as discussed above, \myinline{assume} may simply add the constraint to the solver state without calling the solver immediately. 
It would then be the responsibility of the \myinline{run} function to check the satisfiability of the path condition when reaching a leaf of the SE tree, if not all constraints have been checked yet.

\paragraph{The Symex Implementation} 
The \myinline{Symex} monad should model non-determinism due to branching and attach a state, the path condition, to each branch of the execution.

\begin{wrapfigure}{r}{0.43\textwidth}
\vspace{-1em}
\begin{minipage}
{0.43\textwidth}
\begin{lstlisting}
let state = ref (Solver.create ())  
type 'a t = ('a -> unit) -> unit
       (* = 'a Iter.t *)
\end{lstlisting}
\end{minipage}
\vspace{-1em}
\end{wrapfigure}
In practice, there are several ways of implementing both non-determinism and state in \ocaml. We make two design choices to improve efficiency.
First, we utilise mutable states in \ocaml and implement the path condition as a mutable reference to a \myinline{Solver.t}. 
Second, we opt for \emph{depth-first} exploration of branches, allowing us to leverage the incremental reasoning capabilities of \sot presented above and to optimise memory sharing between branches. 
We do this by implementing the non-determinism monad using idiomatic \ocaml iterators, such that a symbolic computation is simply an iterator over the leaves of the SE tree. 
Specifically, a symbolic computation is a function that takes a continuation \myinline{f : 'a -> unit} and applies it to all results of type \myinline{'a} produced by the computation.
%

\paragraph{Implementing \myinline{branch_on}} 
\label{par:branch_on}
With our design choices above, implementing \myinline{branch_on} is straightforward. This primitive receives a guard, two symbolic computations for the `then' and `else' branches, and a continuation \myinline{f}. 
First, the guard is simplified using the solver's contextual simplification procedure; if simplified to \myinline{true} or \myinline{false}, we immediately execute the corresponding branch.
Otherwise,  we \begin{enumerate*}
\item save the current solver state;
\item add the guard to the path condition;
\item check the satisfiability of the path condition, and if satisfiable, we pass the continuation \myinline{f} to the `then' branch for execution;
\item backtrack the solver state to the saved state;
\item add the negation of the guard to the path condition,  and (analogously) if satisfiable, we pass \myinline{f} to the `else' branch for execution.
\end{enumerate*}
We present our implementation in detail in \ifbool{extendedversion}{appendix \cref{app:branching-impl}}{the extended version~\cite{extended-version}} for the interested reader.

\subsection{The Data Module}\label{sec:data}

We give a high-level overview of \myinline{Soteria.Data} as an example of a module that provides reusable components on top of the \myinline{Symex} interface. 
The \myinline{Data} module contains data structures that have been lifted to the symbolic world, together with operations lifted to symbolic computations.

\paragraph{Maps}
We focus on \myinline{Map}, a key-value map data structure which exists in the \ocaml standard library, but could not be used soundly with symbolic keys out of the box and requires a custom symbolic implementation. For instance, consider the \myinline{find_opt} operation that retrieves the value associated with a key in a map, and returns \myinline{None} if the key is absent.
\begin{lstlisting}
(* OCaml standard library *)
val find_opt : 'a Stdlib.Map.t  -> Key.t -> 'a option

(* Soteria.Data *)
val find_opt : 'a Soteria.Map.t -> Key.t -> 'a option Symex.t
\end{lstlisting}
  
Let us explain why the \lstinline{Stdlib} implementation of \myinline{find_opt} is unsound when used with symbolic keys, and how we address this issue in our \lstinline{Soteria.Data.Map} implementation. When retrieving the value associated with a key, the \lstinline{Stdlib} implementation finds an entry in the map whose key is \emph{syntactically} equal to the queried key (the same \ocaml value).
However, a symbolic computation may add a map entry for \emph{symbolic} key $\symb{x}$ (where $\symb{x}$ is a symbolic variable) and later attempt to retrieve the value of symbolic key $\symb{y}$ (a different variable, and hence a different \ocaml object). 
As $\symb{x}$ and $\symb{y}$ are different \emph{syntactically}, \myinline{find_opt} would return \myinline{None}.
However, this is unsound under path condition $\symb{x} = \symb{y}$: the two symbolic keys are equal, and thus the entry for $\symb{x}$ should be returned when querying for $\symb{y}$.
On the other hand, \myinline{Soteria.Data.Map} uses \emph{symbolic equality} (exposed by the \myinline{SymEq} module type) on keys and returns a symbolic computation (\myinline{Symex.t}) that queries the path condition to determine if two keys are equal. Executing this symbolic computation under a path condition that does not constrain the two key variable returns \emph{two branches}: one where the keys are equal and the entry is returned, and one where they are not and \myinline{None} is returned; in each branch, the path condition is updated with the corresponding constraint (\myinline{x = y} or \myinline{x != y}).

\begin{wrapfigure}{r}{0.43\textwidth}
\vspace{-1em}
\begin{minipage}{0.43\textwidth}
\begin{lstlisting}
let find_opt_sym map key =
  let rec find_bindings = function
    | [] -> Symex.return None
    | (k, v) :: tl ->
        if%sat Key.eq key k
        then Symex.return (Some v)
        else find_bindings tl
  in
  (* Syntactic lookup *)
  match M.find_opt key st with
  | Some v -> Symex.return (Some v)
  | None ->
    find_bindings (M.to_list map)
\end{lstlisting}
\end{minipage}
\vspace{-2em}
\end{wrapfigure}
\paragraph{Branching and Optimisation} In theory, this symbolic implementation of \myinline{find_opt} (on the right) may lead to an exponential blowup in the number of branches when querying for a key in a map with many entries, as each entry may be equal to the queried key or not. In practice, however, that is not the case, thanks to \begin{enumerate*}\item the optimisations described earlier in this section, \item constraints that are already in the path conditions when accessing such maps, and \item additional optimisations we have implemented in the \myinline{Data.Map} module\end{enumerate*}. In particular, when querying for a key, we first check if the key is syntactically equal to any of the keys in the map, which is a cheap check that may avoid branching altogether.

In general, this implementation illustrates how a call to \myinline{if\%sat} does not necessarily lead to branching. In fact, in the \rusteria example we describe later in \cref{sec:discussions} (where values are inserted into a \myinline{BTreeSet}), we record more than 300M calls to \myinline{if\%sat}, of which only 4683 branches are feasible.

\paragraph{Other Data Structures} In theory, any data structure from the \ocaml standard library can be adapted in this way and used soundly in \sot, so long as any of its operations that manipulate symbolic values are lifted to sound symbolic computations using the \myinline{Symex} interface.
For instance, we also provide a \myinline{Range} module to represent pairs of symbolic integers and reason about their ordering and inclusions.

\paragraph{Parametricity}
\myinline{Soteria.Data} also provides a series of module types (\ocaml's analogue of interfaces/typeclasses) for symbolic abstractions. For instance, it provides an interface for users to describe ``an integer'' or ``values that can be compared for equality'' (used in the \myinline{Map} implementation to characterise the type of keys). This is particularly useful since \sot leaves the choice of the representation of symbolic values to the user, and thus the \myinline{Data} module cannot assume \eg a specific representation of symbolic integers.
Note that the data structures in \myinline{Data} are parametrised by the choice of an SE monad \myinline{Symex} (itself parametrised by a solver and symbolic values).
As such, users may swap the choice of values, solvers, or even the order of branch exploration, \emph{without} having to reimplement these data structures.

\section{\rusteria: \sot for Rust}\label{sec:soteria-rust}

\begin{wrapfigure}[7]{r}{0.5\textwidth}
  \vspace{-10pt}
  \includegraphics[width=0.5\textwidth]{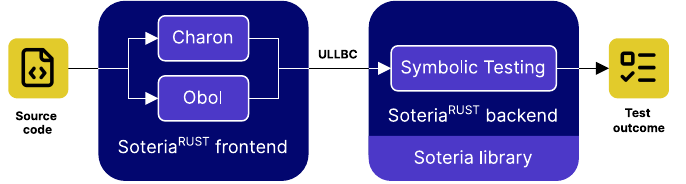}\vspace{-5pt}
  \caption{The \rusteria architecture}\label{fig:rust-pipeline}
\end{wrapfigure}
We instantiate \sot to develop \rusteria, a \emph{symbolic execution (SE) engine for Rust}.
As we show below, the performance of \rusteria, both in terms of speed and the bugs detected, is \emph{comparable or often better} than the state of the art tool Kani~\cite{how-open-source-thekaniteam-2023} and passes a very large fragment of the tests for Miri~\cite{miri-rust-team}, the reference (concrete) interpreter for Rust.
\rusteria targets more Rust features than Kani: unlike Kani, it accounts for \emph{\TBs}~\cite{tree-borrows-villani-2025}, the aliasing model of Rust.
Unlike Miri, \rusteria performs \emph{symbolic} rather than \emph{concrete} testing.
To our knowledge, \rusteria is the \emph{first} SE engine that supports \TBs.
We proceed with a description of the \rusteria engine (\cref{subsec:rusteria_engine}) followed by an in-depth evaluation of its performance (\cref{subsec:rusteria_evaluation}).

\subsection{The \rusteria Engine}
\label{subsec:rusteria_engine}
\rusteria operates on ULLBC (Unstructured Low-Level Borrow Calculus) code~\cite{aeneas-rust-verification-ho-2022}, a representation of MIR (Mid-level Intermediate Representation) designed with formal analysis in mind.

\begin{wrapfigure}{r}{0.46\textwidth}
\vspace{-1em}
\begin{lstlisting}
let exec_stmt stmt =
  match stmt.kind with
  | Nop -> ok ()
  | Assign (place, rval) ->
    let* ptr= resolve_place place in
    let* v = eval_rvalue rval in
    State.store ptr place.ty v
\end{lstlisting}
\vspace{-1.5em}
\end{wrapfigure}
As shown in \cref{fig:rust-pipeline}, \rusteria has two frontends for parsing Rust:
\begin{enumerate*}
  \item \emph{Obol}, our own front-end which directly hooks into the Rust compiler public API to extract ULLBC, and is tailored for efficiently extracting monomorphised code; and
	\item Charon~\cite{ho2025charonanalysisframeworkrust}, which is currently slower and supports a smaller surface of the Rust language, but supports polymorphic code and is therefore suitable for our planned future work on enabling polymorphic symbolic execution in \rusteria.
\end{enumerate*}
We also provide a Kani-compatible API, allowing users to run \rusteria as a drop-in replacement for Kani~\cite{how-open-source-thekaniteam-2023} and run (existing) Kani tests with \rusteria.

We design symbolic values by lifting \myinline{Bv_values} to Rust as defined in the \ifbool{extendedversion}{appendix (\cref{fig:rust-values} in \cref{app:rusteria})}{extended version~\cite{extended-version}}.
The \rusteria interpreter straightforwardly traverses the ULLBC AST from the found entry points (a \myinline{main} function, or functions annotated with \myinline{#[kani::proof]}), executing statements and evaluating expressions.
The execution is wrapped inside a state monad, itself wrapped in the \lstinline|Symex| monad of \sot, allowing the state and variable store to be threaded through SE.
The lifting above shows an excerpt of the interpreter for executing the \myinline{Nop} and \myinline{Assign} statements.

\paragraph{The \rusteria State}
The \rusteria interpreter is parametric on a module that implements the Rust \lstinline|State| interface, which defines the type of the state and pointers, as well as symbolic computations for interacting with the state using these pointers.

\ifbool{extendedversion}{\ifbool{reviewversion}{\newpage}{}}{}

\begin{wrapfigure}{r}{0.41\textwidth}
\vspace{-1em}
\begin{lstlisting}
type ptr = {
  ptr   : BV_Value.(ptr t);
  tag   : Tree_borrow.tag;
  align : BV_Value.(nonzero t);
  size  : BV_value.(int t);
}
type state = (* simplified *)
(loc, Tree_block.t * Tree_bor.t)
  Soteria.Data.Map.t
\end{lstlisting}
\vspace{-1.5em}
\end{wrapfigure}
Currently, \rusteria provides one state implementation.
Its pointers, defined across, are structures comprising a \myinline{BV_Value} pointer (a pair of symbolic location and offset, in the style of CompCert~\cite{compcert-memory-model-leroy-2012}), a \TB tag (all \TB-related terms are explained below), and symbolic values corresponding to the alignment and size of the allocation the pointer points to.
Keeping track of the alignment and size within the pointer itself enables us to efficiently check for various kinds of pointer arithmetic mistakes, without repeatedly retrieving that information from the state.
States (simplified for presentation) are symbolic maps (implemented using \myinline{Soteria.Data.Map}) from symbolic locations to memory objects.
Each memory object is a pair comprising a \emph{tree block} and a \TB.
A tree block (unrelated to \TB) is a data structure for efficient and automated reasoning about symbolic byte arrays, following the design of \citewithauthor{gillian-foundations-implementation-ayoun-2024}, and extended to annotate each byte with a \TB state.
The implementation of this state was greatly simplified by the data structures provided by \sot and the ability to define our own \ocaml data structures and types such as \TBs, \TB blocks and \TB tags.



\paragraph{\TBs}
%
%
%
%
With these data structures in place, implementing \TBs is surprisingly simple.
In the \TB model~\cite{tree-borrows-villani-2025}, reading from or writing to a byte with \TB state $s$ through a pointer with tag $t$ requires taking a transition in a state machine depending on $s$, $t$, and the structure of the tree $T$ of the \TB associated with the allocation.
We can do this in \ocaml with a single a pattern match over the state $s$ and the kind of transition (determined by $t$ and $T$) to determine the next state $s'$, when the operation is allowed;
otherwise (when it is not allowed), we raise an error.
By contrast, had we compiled Rust to an IL, we would have had to encode this entire logic in the IL as well, significantly increasing the complexity of its implementation.

\begin{wrapfigure}{r}{0.455\textwidth}
\vspace{-1em}
\begin{lstlisting}[numbers=left,xleftmargin=1.5em,escapechar=^]
fn main() {
  let mut root = 42;
  let ptr = &mut root as *mut i32;
  let (x, y) = unsafe {
    (&mut *ptr, &mut *ptr) };^\label{line:tb-create-x-y}^
  *x = 13;^\label{line:tb-use-x}^
  *y = 20; // UB: y is disabled^\label{line:tb-use-y-ub}^
}
\end{lstlisting}
\vspace{-1.5em}
\end{wrapfigure}
Consider the example across by \citewithauthor{tree-borrows-villani-2025}: two mutable references \myinline{x} and \myinline{y} are created on line~\ref{line:tb-create-x-y},  aliasing the same address.
This is forbidden in safe Rust, but bypassed here using raw pointers and the \myinline{unsafe} keyword.
When \myinline{x} is written to on line~\ref{line:tb-use-x}, \TBs dictate that the state of each modified byte is updated to \emph{disable} access through all other tags, including that of \myinline{y}.
As such, using \myinline{y} on line~\ref{line:tb-use-y-ub} triggers UB, which \rusteria correctly detects by checking the state of the tag associated with \myinline{y}.

\begin{wrapfigure}{r}{0.57\textwidth}
 \vspace{-1em}
\begin{lstlisting}[language=Rust,style=ruststyle]
#[rustc_intrinsic]
pub const unsafe fn copy_nonoverlapping<T>(
    src: *const T, dst: *mut T, count: usize)
\end{lstlisting}%
 \vspace{-0.3em}
\begin{lstlisting}[language=OCaml,style=ocamlstyle]
val copy_nonoverlapping :
  t:ty -> src:ptr -> dst:ptr ->
  count:BV_value.(int t) -> unit ret
\end{lstlisting}%
\vspace{-1.2em}%
\end{wrapfigure}
\paragraph{Intrinsics}
When designing SE engines, one common challenge is the need to support numerous functions that are built into the language in addition to all language features.
Rust is no exception and defines, at the time of writing~\cite{rust-intrinsics-2025}, a total of 223 \emph{intrinsics}.
These functions allow for all kinds of low-level operations such as floating point arithmetic or atomic manipulation.
Unfortunately, these intrinsics are rather unstable, which could constitute a maintenance burden over time without careful design.

To address this challenge in \rusteria, we use the Rust signatures of these intrinsics to automatically generate the corresponding \ocaml signatures, together with all the necessary boilerplate to integrate them into the interpreter.
We then lift these \ocaml signatures to SE computations, receiving symbolic arguments and returning symbolic computations (hidden behind the \myinline{ret} type).
This design ensures that the signatures are always in sync with their ever-evolving Rust definitions, and streamlines the addition of each new intrinsic.

Note that the developer still needs to implement the actual intrinsic itself, matching the generated interface. In most cases, the implementation is straightforward: the average implementation is 5.6 lines of \ocaml code across 165 implemented intrinsics, with the longest implementation being 35 lines of code. The challenge of handling intrinsics thus lies in keeping up to date with their ever-evolving signatures (which our design addresses) rather than in their implementation.

%
%
%


\subsection{Evaluation}
\label{subsec:rusteria_evaluation}
We evaluate \rusteria against Kani and Miri.
Miri is a concrete interpreter for Rust and can detect undefined behaviours (UBs).
It is the \textit{de facto} tool for testing Rust code for UBs and is used extensively by the Rust compiler team.
Both Kani and Miri come with their own test suites, and we thus evaluate all three tools on both Kani and Miri suites (\cref{subsubsec:kani_miri_suites}).
We further compare all three tools on a fresh test suite that is not biased towards either tool (\cref{subsubsec:ds_suite}).
Finally, we compare \rusteria and Kani on the tests of a real-world Rust library, \myinline{finetime}~\cite{finetime-van-woerkom-2025} (\cref{subsubsec:finetime_case_study}).
We run all tests on a MacBook Pro (M4 Pro, 14 cores, 32GB RAM).

\subsubsection{The Kani and Miri Test Suites}
\label{subsubsec:kani_miri_suites}
The two suites comprise 922 tests in total: 413 Kani (commit \texttt{6c2c4f0}) and 509 Miri (commit \texttt{ec775c1}) tests.
For the Kani suite, we run Kani as intended (with \myinline{//kani-flags} annotations in the tests); we run Miri by replacing  \myinline{#[kani::proof]} annotations with \myinline{#[test]} and assuming they are concrete (if any Kani built-in is reached, Miri gives up, as it does not support \emph{symbolic} execution). 
We run \rusteria with \begin{enumerate*}\item unlimited fuel, ensuring \emph{all} feasible paths are explored,\item Kani compatibility enabled, linking against the Kani library wrapper, and \item memory leak and aliasing checks disabled (as Kani does not check for them)\end{enumerate*}.

For the Miri suite, we run Miri as intended, and we run \rusteria as is, since Miri compatibility is built-in.
We run Kani with flags \myinline{uninit-checks} and \myinline{valid-value-checks} to enable uninitialised memory access and validity checks, ensuring a similar level of thoroughness to Miri and \rusteria.

Kani (\resp Miri) takes 18s (\resp 4s) to run its longest test, so we set a timeout of 23s (\resp 9s) for all three tools on the Kani (\resp Miri) suite.
%
We classify the test outcomes into four categories: \emph{pass} (where the test outcome matches the expected outcome), \emph{fail} (when the test outcome disagrees with the expected outcome), \emph{unsupported} (if the tool crashes or raises an error, or a feature is not supported), and \emph{timeout}.

\begin{figure}[t]
\centering\small
\centering\footnotesize
\renewcommand{\arraystretch}{1.05}
\begin{tabular}{l@{\hspace{2pt}}@{\hspace{2pt}}rrrr@{\hspace{12pt}}rrrr@{\hspace{12pt}}c}
\toprule
\multirow{2}{*}{Tool} & \multicolumn{4}{c}{Kani Suite} & \multicolumn{4}{c}{Miri Suite} & Total \\
& \#Pass & \#Fail & \#Unsup. & \#Timeout & \#Pass & \#Fail & \#Unsup. & \#Timeout & \% Pass \\ \midrule
Kani      & 413 & 0  & 0   & 0  & 207 & 106 & 190 & 6 & \textbf{67.2\%}  \\
Miri      & 236 & 8  & 169 & 0  & 509 & 0   & 0   & 0 & \textbf{80.8\%}  \\
\rusteria & 381 & 18 & 9   & 5  & 420 & 61  & 14  & 14 & \textbf{86.9\%}  \\ \bottomrule
\end{tabular} \vspace{10pt}\\
\hrule height 1pt \vspace{5pt}
\begin{subfigure}[b]{0.49\textwidth}
    \includegraphics[width=\textwidth]{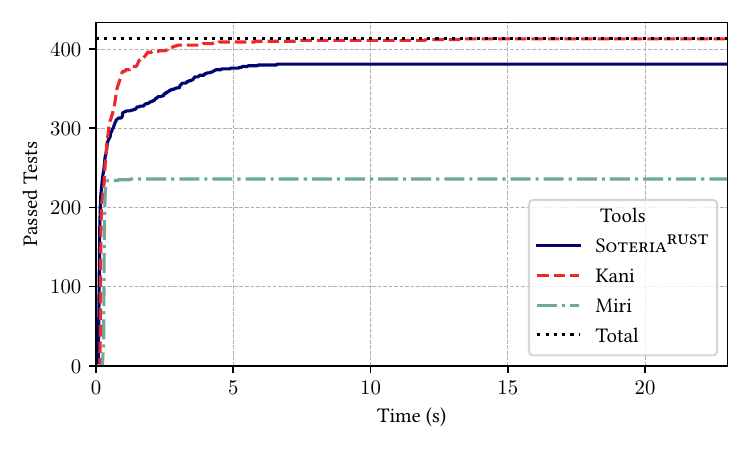}
    \vspace{-20pt}
    \caption{Kani suite}\label{fig:survival-graph-kani}
  \end{subfigure}
  \hfill
  \begin{subfigure}[b]{0.49\textwidth}
    \includegraphics[width=\textwidth]{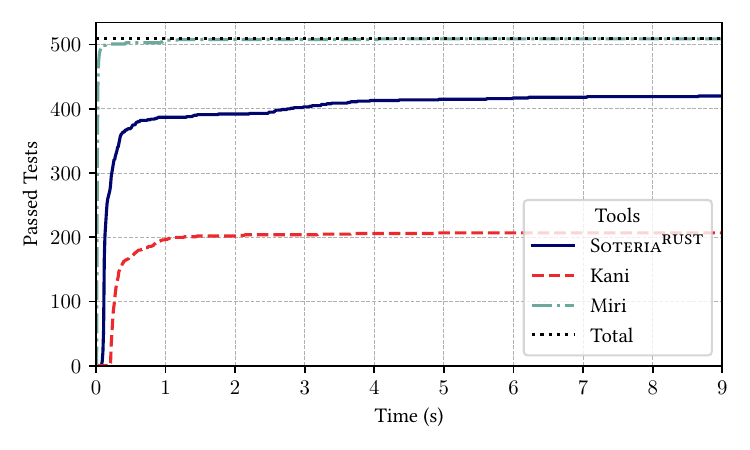}
     \vspace{-20pt}
    \caption{Miri suite}\label{fig:survival-graph-miri}
  \end{subfigure} \vspace{-5pt}
\caption{The Kani and Miri test suite results (above); number of tests passed over time by each tool (below)}
\label{tab:benchmark-res}
\label{fig:survival-graphs}
\end{figure}


We present the results for each test suite in \cref{tab:benchmark-res} (above), as well as the number of tests passed by each tool over time (survival graph) in \cref{fig:survival-graph-kani,fig:survival-graph-miri}.
Unsurprisingly, both Kani and Miri pass all tests in their respective suites;
for both suites, \rusteria performs the second best and far outperforms Kani (\resp Miri) on the Miri (\resp Kani) suite.
Moreover, \rusteria passes the highest number of tests across both suites (86.9\%).
\rusteria does not support 23 tests in total across both suites, due to gaps in the frontend (Obol does not support \myinline{TypeId}), and to unsupported features in the engine (\myinline{&dyn Trait} method calls where the first argument is \myinline{Self}, as it is unsized).
Additionally, \rusteria fails 79 tests, \ie it either detects an error where there is none, or misses an error.
These failures, which predominantly happen in the Miri suite, are due to UB behaviours not being checked, \eg not checking for address overlap between function arguments and return address, or not checking the validity of \myinline{&dyn Trait} upcasting. The failures in Kani are mostly due to trigonometric intrinsics (such as \myinline{cosf32} or \myinline{log10f32}) being implemented with insufficient precision, or to limitations around float operations (SMT-LIB does not handle NaN patterns accurately).

Thanks to \sot, a \emph{first-year PhD student} developed a competitive SE engine for Rust in merely eight person-months.
The limitations of Miri are due to its concrete nature and are hardly fixable, and those of Kani are due to its incomplete support for Rust semantics, which though fixable, require a significant engineering effort.
By contrast the limitations of \rusteria are due to missing features that can be added with little friction, given the modular design of \sot.

\begin{figure}[t]
\centering
\begin{subfigure}[b]{0.48\textwidth}
\begin{lstlisting}[language=Rust,style=ruststyle]
let x: u128 = kani::any();
let mut count: u32 = 0;
for i in 0..u128::BITS {
  let bit = x & (1 << i);
  if bit != 0 { count += 1; }
}

assert!(count == ctpop(x));
\end{lstlisting}\vspace{-5pt}
\caption{Original test from the Kani suite}\label{fig:ctpop-test}
\end{subfigure}
\hfill
\begin{subfigure}[b]{0.48\textwidth}
\begin{lstlisting}[language=Rust,style=ruststyle]
let x: u128 = kani::any();
let mut count: u32 = 0;
for i in 0..u128::BITS {
  let bit = x & (1 << i);
  count = count.wrapping_add(
    (bit != 0) as u32);
}
assert!(count == ctpop(x));
\end{lstlisting}\vspace{-5pt}
\caption{Modified test, verifiable by \rusteria}\label{fig:ctpop-test-fixed}
\end{subfigure}
\vspace{-5pt}
\caption{Example of a biased test in the Kani test suite, and a modification to make it verifiable by \rusteria}\label{fig:ctpop-tests}
\vspace{-10pt}
\end{figure}

\paragraph{Test Bias}
We note that both Kani and Miri suites are biased towards their respective tools.
For instance, consider the test for the \myinline{ctpop} intrinsic in the Kani suite, shown in \cref{fig:ctpop-test}, which counts the number of 1 bits in a given number \myinline{x}.
Both Kani and \rusteria pass this test, but the test is biased towards the strengths of Kani.
Specifically, despite the branch (if-condition) in the loop, CBMC easily merges the branches into a single path.
By contrast, \rusteria currently has no such branch merging capacity and branches on every loop iteration, leading to $2^{128}$ paths for a 128-bit integer.
However, as shown in \cref{fig:ctpop-test-fixed}, a simple modification (where we drop the `if' and use \texttt{wrapping\_add} to avoid checking for overflow) eliminates branching and allows \rusteria to verify the property in significantly less time.
\rusteria runs the modified test faster than Kani:
Kani passes the first test in 24.5s and the modified test in 14.7s, while \rusteria times out on the first test and passes the modified test in 1.1s.

Similarly, as we show in \cref{subsubsec:ds_suite}, code patterns such as heap-intensive code and pointer manipulation (which are ubiquitous in real-world Rust code) degrade the performance of Kani significantly.
Our case study in \cref{subsubsec:ds_suite} suggests that \rusteria handles these patterns more efficiently. We discuss path-merging versus path-exploration more extensively in \cref{sec:discussions}.

\subsubsection{Data Structures Tests (DS Suite)}
\label{subsubsec:ds_suite}
To showcase the performance of \rusteria on real, unbiased code, we wrote a number of symbolic tests for several Rust data structures, including \myinline{Vec}, \myinline{LinkedList}, \myinline{BTreeMap}, \myinline{BinaryHeap}, \myinline{VecDeque}, \myinline{Option}, and \myinline{Result}.
The tests (72 in total) create data structures with an arbitrary size up to a bound and perform a number of operations such as inserting, removing, and accessing symbolic elements.
We present their results in \cref{tab:datastructure-tests}, run with a 10-second timeout, and their survival graph in \cref{fig:std-survival-graph}.
As these tests are inherently symbolic, we did not run Miri on them (Miri cannot run symbolic tests).

\begin{figure}[t]
\centering\small
\centering\footnotesize
\setlength{\tabcolsep}{8pt}
\renewcommand{\arraystretch}{1.05}
\begin{tabular}{l @{\hspace{2pt}}@{\hspace{2pt}}ccccccc@{\hspace{2pt}}@{\hspace{2pt}}c}
\toprule
Tool & \lstinline|BinaryHeap| & \lstinline|BTreeMap| & \lstinline|LinkedList| & \lstinline|Option| & \lstinline|Result| & \lstinline|Vec| & \lstinline|VecDeque| & Total\!\! \\ \midrule
Kani & 1 & 1 & 0 & 14 & 16 & 8 & 1 & \textbf{41} \\
Kani (unwind.) & 4 & 2 & 7 & 15 & 16 & 9 & 5 & \textbf{58} \\
\rusteria & 6 & 8 & 8 & 15 & 16 & 11 & 8 & \textbf{72} \\
Total Tests & \textbf{6} & \textbf{8} & \textbf{8} & \textbf{15} & \textbf{16} & \textbf{11} & \textbf{8} & \textbf{72} \\
\bottomrule
\end{tabular}
\vspace{10pt}\\
\hrule height 1pt \vspace{5pt}
\begin{subfigure}[b]{0.49\textwidth}
    \includegraphics[width=\textwidth]{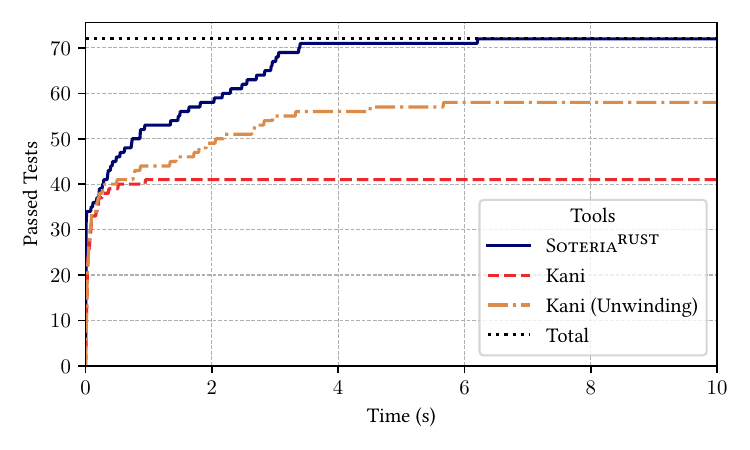}
    \vspace{-20pt}
    \caption{Tests passed in the DS suite over time}\label{fig:std-survival-graph}
  \end{subfigure}
  \hfill
  \begin{subfigure}[b]{0.49\textwidth}
      \includegraphics[width=\textwidth]{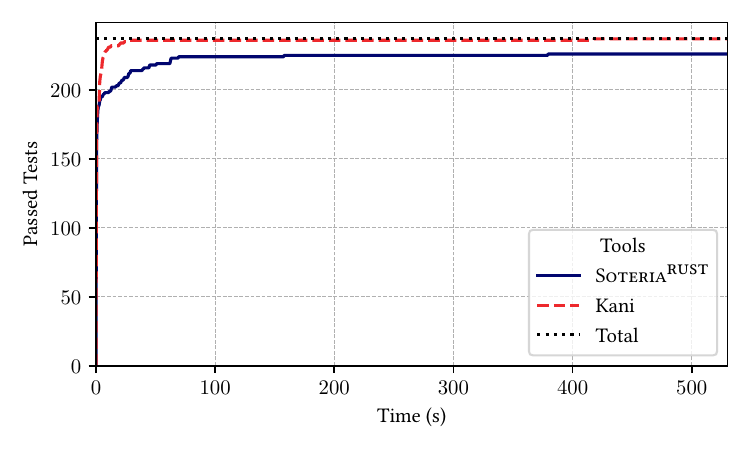}
      \vspace{-20pt}
      \caption{Tests passed in the \texttt{finetime} crate over time}
      \label{fig:finetime-survival-graph}
  \end{subfigure}\vspace{-5pt}
\caption{The DS suite results with a timeout of 10s (above); tests passed over time by each tool (below).
\vspace{-10pt}}
\label{tab:datastructure-tests}
\label{fig:survival-graphs}
\end{figure}


While \rusteria passes \emph{all} tests in 6.6 seconds, Kani only passes 41 (\resp 58) without (\resp with) loop unwinding annotations (\myinline{#[kani::unwind(N)]}); \rusteria ignores these annotations.
We use the minimum unwinding value that ensures no branches are missed.
Furthermore, we run Kani without validity and undefined memory access checks; with these enabled, the tool gives up on the majority of tests, due to incomplete support of uninitialised memory checks.
We run Kani using CaDicaL~\cite{cadical-2-computationally-efficient-biere-2024} as the underlying SAT solver; we observed no notable difference when using Kissat~\cite{entering-sat-competition-biere-2024} instead.
As shown, \rusteria is faster on DS tests, while performing a deeper analysis with validity, undefined memory access, \emph{and} \TB checks.

We highlight that these tests were written \emph{with no knowledge of Kani limitations} and \emph{with no special accommodations to advantage \rusteria}.
In particular, to eliminate bias towards \rusteria, we removed ourselves from the test writing process and used an AI pair programmer to generate them; we present the AI prompt we used in \ifbool{extendedversion}{appendix \cref{app:rusteria-tests-prompt}}{the extended version~\cite{extended-version}}.
The only adjustment we made was to the size bounds of the data structures to ensure the tests complete in a reasonable time.

\subsubsection{Real World Usage}
\label{subsubsec:finetime_case_study}

We compare \rusteria with Kani on {\myinline{finetime}~\cite{finetime-van-woerkom-2025}}, an efficient and high-precision time keeping library.
We used the version from commit \texttt{f5962b8}.
The crate has 237 symbolic tests using the Kani interface, which test the panic-freedom of various calendar conversions, the accuracy of rounding durations, and the correctness of round-trip conversions between time representations.
We note that the \myinline{finetime} tests were written for Kani, and they may thus be biased towards its strengths.

We present the results in \cref{fig:finetime-survival-graph}.
Neither \rusteria nor Kani fail; Kani passes the slowest test in 515s, while \rusteria times out on 11 tests (with timeout set to 530s).
In all 11 cases, \rusteria gets stuck waiting for the solver (Z3~\cite{z3-efficient-smt-demoura-2008}) to resolve queries involving signed remainder and division operations on bitvectors, which are a known limitation of SMT solvers.

Despite limitations around signed operations, \rusteria is competitive with Kani on real-world code, even when bit operations are heavily used.
We plan to improve the support for these operations in future work, with formally verified reductions to simpler operations~\cite{certified-decision-procedures-width-independent-bhat-2025}.

We have also used an AI agent to write symbolic tests for a variety of real-world Rust libraries. This allowed us to find a case of potential UB in the \myinline{hashbrown} crate~\cite{hashbrown-2025}, which implements hash maps and sets, and is used in the Rust standard library; it is also the second most downloaded Rust crate~\cite{crates-io-stats}. We submitted a pull request to fix the issue, which has since been merged.\footnote{\ifbool{reviewversion}{Link to the bug report redacted for anonymity}{\url{https://github.com/rust-lang/hashbrown/pull/692}}}
\section{\coteria: \sot for C}\label{sec:soteria-c}

We develop \coteria, a symbolic execution (SE) engine for C that can perform \emph{both} whole program symbolic testing (WPST) and fully automated, bi-abductive bug detection.
We present an overview of its infrastructure and evaluate its performance and bug-finding capabilities against existing tools.

\subsection{Architecture and Engine}

We use Cerberus~\cite{cerberus-semantics-memarian-2023} to parse C and de-sugar it into AIL, which is close to the abstract C source, but with all types checked and canonicalised.
Much like \rusteria, \coteria is defined by recursively traversing the source AST.
The two core functions of the engine, \myinline{eval_stmt} (evaluating a C statement) and \myinline{eval_expr} (evaluating an expression) return an \myinline{InterM.t} monad value, which is obtained by applying the state monad combinator to the \myinline{Symex.Result} monad.
Intuitively, \myinline{InterM.t} is an SE monad where each branch may either be successful or erroneous, and where each branch carries an additional state corresponding to the C heap.
For example, evaluating \myinline[language=C]{if} is as shown below.
\begin{wrapfigure}[8]{r}{0.5\textwidth}
\begin{lstlisting}
let rec exec_stmt = function
| AilSif (cond, then_stmt, else_stmt) ->
  let* v = eval_expr cond in
  let* v = cast_aggregate_to_bool v in
  if%sat v then exec_stmt then_stmt
  else exec_stmt else_stmt
| ...
\end{lstlisting}
\vspace{-10pt}
\end{wrapfigure}
\indent First, the guard \myinline{cond} is evaluated, obtaining the symbolic value \myinline{v}. C expressions evaluate to an `aggregate' value, which is either a base value (integer, float or pointer) or a structure; therefore, \myinline{v} must be cast to a symbolic boolean using \myinline{cast_aggregate_to_bool}. Finally, we call \myinline{if\%sat} to evaluate the `then' and `else' statements depending on the satisfiability of \myinline{v}.

\paragraph{Bi-abduction and Bug Finding} We implement `fix-from-error' bi-abduction, adapting existing work~\cite{javert-20-compositional-fragososantos-2019,compositional-symbolic-execution-loow-2024,gillian-foundations-implementation-ayoun-2024} to \sot.
The main novelty of our approach lies in formulating the transformation from standard SE to bi-abductive SE as a state monad transformer, which we distribute as part of the \sot library in the \myinline{State} module. We elide the details of this transformer for lack of space, though it is defined and proven in~\citewithauthor{gillian-foundations-implementation-ayoun-2024}.

As with \pulse~\cite{finding-real-bugs-le-2022}, automatic bug finding works by first running bi-abductive SE to generate function summaries, and then applying a `manifest bug' criterion to the summaries that have an erroneous post-condition to determine whether a bug should be reported to the user.

\subsection{Evaluation}
\label{sec:evaluation-c}
\begin{figure}[t]
  \centering\small
\centering\footnotesize
\setlength{\tabcolsep}{8pt}
\renewcommand{\arraystretch}{1.05}
\begin{tabular}{@{} l r r r r @{}}
\toprule
Folder & Tests & CBMC (s) & Gillian-C (s) & \coteria (s) \\
\midrule
array     & 21 & 131.45 & 15.43 & \textbf{11.39} \\
deque     & 34 & 278.47 & 23.89 & \textbf{11.41} \\
list      & 37 & 134.94 & 25.85 & \textbf{23.73} \\
pqueue    & 2  & 2.27   & 1.42  & \textbf{1.12} \\
queue     & 4  & 4.33   & 2.73  & \textbf{1.42} \\
rbuf      & 3  & 3.13   & 2.08  & \textbf{1.00} \\
slist     & 37 & 48.03  & 25.40 & \textbf{15.28} \\
stack     & 2  & 2.16   & 1.36  & \textbf{0.77} \\
treeset   & 6  & 180.00 & 4.28  & \textbf{2.85} \\
treetable & 13 & 390.00 & 8.97  & \textbf{5.57} \\
\midrule
Total     & 159 & 1174.78 & 111.41 & \textbf{74.54}  \\
\bottomrule
\end{tabular}
  \caption{Whole-program symbolic testing (WPST) results on Collections-C with shortest times \textbf{highlighted}.}
  \label{tbl:wpst-collections-c}
  \vspace{-10pt}
\end{figure}
We evaluate \coteria on Collections-C, used by Gillian-C to evaluate WPST and bi-abduction~\cite{gillian-foundations-implementation-ayoun-2024,gillian-part-multilanguage-fragososantos-2020,compositional-symbolic-execution-loow-2024}.
We run all benchmarks on a MacBook Pro (M4 Pro, 14 cores, 32GB RAM).

\paragraph{Whole-Program Symbolic Testing (WPST)} We run \coteria, Gillian-C and CBMC on the same 159 symbolic tests as \citewithauthor{gillian-part-multilanguage-fragososantos-2020}, with a timeout of 30s per test, and compare their results.
We run the three tools with options that, we believe, provide a fair comparison with \coteria in terms of number of checks.%
\footnote{Specifically, we run Gillian-C with default options, 
and we run \coteria with \lstinline[language=bash]|--alloc-cannot-fail --infinite-fuel --havoc-undef --cbmc|.
We run CBMC~6.8.0 (Nov 5th 2025) with the \lstinline[language=bash]|--bounds-check --pointer-check --div-by-zero-check --pointer-primitive-check --unwind 10 --drop-unused-functions --signed-overflow-check --pointer-overflow-check --float-overflow-check --unwinding-assertions --no-malloc-may-fail|. 
} In particular, all three tools are run such that there is no limit on the number of explored paths, meaning that \emph{all} feasible branches are explored (the tests are written to be bounded).
We present the performance results in \cref{tbl:wpst-collections-c}.

\begin{wrapfigure}[7]{r}{0.5\textwidth}
\vspace{-10pt}
\centering
\begin{minipage}{0.5\textwidth}
\label{tab:tool_status_summary}
\centering\footnotesize
\setlength{\tabcolsep}{8pt}
\renewcommand{\arraystretch}{1.05}
\begin{tabular}{@{} l @{\hspace{10pt}} c @{\hspace{10pt}} c @{\hspace{10pt}} c @{\hspace{10pt}} c @{\hspace{10pt}} c @{}}
\toprule
Tool & Pass & Fail (FN) & Fail (FP) & Crash & Timeout \\
\midrule
\coteria  & \textbf{159} & \textbf{0} & \textbf{0} & \textbf{0} & \textbf{0} \\
Gillian-C & 156 & 3 & \textbf{0} & \textbf{0} & \textbf{0} \\
CBMC  & 124 & 2 & 2 & 1 & 30 \\
\bottomrule
\end{tabular}
\end{minipage}\vspace{-5pt}
\caption{The WPST result categories on Collections-C}
\label{fig:wpst_categories}
\end{wrapfigure}%
Note that these tests were written by the Gillian-C authors, and thus they are subject to \emph{bias} (towards Gillian-C), as we described above in \cref{sec:soteria-rust}. 
\coteria runs \emph{faster} than Gillian-C on \emph{all tests}; this is despite the fact that, unlike \coteria, Gillian-C does not perform integer overflow checks. 
Both \coteria and Gillian-C run \emph{significantly faster} than CBMC.
However, as these tests are (potentially) biased towards Gillian-C,  further evaluation is needed to compare the tools on tests written for \sot and for CBMC.
Moreover, CBMC times out on 30 of the (159) tests (with the timeout of 30s). As such, 900s of the total $\sim$1190s taken by CBMC is spent on timed-out tests.

Moreover, as shown in \cref{fig:wpst_categories}, \coteria performs a \emph{more accurate} analysis than both Gillian-C and CBMC. 
Specifically, all bugs found by Gillian-C and CBMC were also found by \coteria; that is,  to our knowledge \coteria has no false negatives (FNs, missed bugs which were found by at least one other tool). 
On the other hand, Gillian-C and CBMC have three and two FNs, respectively.
In the case of Gillian-C, two of the FNs occur because the Gillian compiler \emph{optimises away} bugs due to accessing uninitialised memory; the reason for the third FN in Gillian-C is unclear to us. 
Similarly, CBMC misses the same two bugs as Gillian-C due to accessing uninitialised memory; this is because CBMC does not check for accesses to uninitialised memory.

Moreover, \coteria and Gillian-C have no \emph{false positives} (FPs), do not crash and do not time out. 
However, as discussed above, CBMC times out on 30 of the tests and crashes on one. 
Lastly, CBMC reports two FPs; \ie CBMC reports two bugs where there are none.
This is because CBMC performs over-approximate (\mox) analysis that is not always precise enough for bug detection.


\ifbool{extendedversion}{}{\ifbool{reviewversion}{\newpage}{}}
\paragraph{Automatic Bug Finding} We evaluate the performance and ability of \coteria to find bugs in Collections-C automatically,  and compare it with \pulse~\cite{finding-real-bugs-le-2022} (version 1.2.0).
%
We run \pulse with \lstinline[language=bash]|--pulse-only -j1| to activate only the \pulse analysis with a single thread.
Both \coteria and \pulse find the same bug, which we have confirmed to be a \emph{true positive}; we submitted a pull request to fix it upstream, which has been accepted.
\coteria is competitive with \pulse in terms of performance.

\begin{wrapfigure}{r}{0.33\textwidth}
\vspace{-1em}
  \centering\small
  \begin{tabular}{@{}lcc@{}}
    \toprule
    Tool & Time (s) & Bugs found \\
    \midrule
    \coteria    & 2.57 & 1 \\
    \pulse  & 2.52 & 1 \\
    \bottomrule
  \end{tabular}\vspace{-1em}
  \caption{Results of bi-abductive bug finding on Collections-C.}
  \label{tbl:bi-abd-collections-c}\vspace{-1em}
\end{wrapfigure}
Thanks to the \myinline{Stats} utility of \sot, we can provide a detailed break-down of the analysis performed by \coteria in 2.57s.  
Specifically, \coteria spends: 0.35s parsing and desugaring the source code; 1.35s analysing the \emph{578 functions} in the codebase (containing assorted control flow) and generating \emph{2364 summaries} (each summary corresponds to one explored satisfiable path); and 0.60s performing ``manifest bug'' analysis on these summaries, isolating the single true positive bug report. The whole process runs through 54{,}925 \myinline{if\%sat} calls, performing 10{,}503 sat checks that are not immediately simplified.

We do not compare \sot with Gillian-C because Gillian-C does not implement a `manifest bug' analysis (needed for identifying true positive bug reports) and only produces function summaries.
Specifically,Gillian-C produces (on Collections-C) over 4000 bug summaries, and thus (without a manifest bug analysis) manual inspection is infeasible.
Moreover, Gillian-C runs significantly slower: Gillian-C takes $\sim$62.8s, while \coteria and \pulse take $\sim$2.57s and $\sim$2.52s, respectively. The data provided in this paragraph for Gillian-C is extracted from \citet{compositional-symbolic-execution-loow-2024}.

In the future we would like to evaluate \coteria on a larger and more diverse set of C benchmarks. Nevertheless, in merely three person-months, we implemented a symbolic C engine that can perform WPST and bi-abductive bug finding. This is thanks to the flexibility and reusability of \sot.
\section{Formalisation}\label{sec:formalisation}
We formalise the symbolic execution (SE) monad that underpins \sot. Our formalisation is closely inspired by existing formalisations~\cite{verified-symbolic-execution-keuchel-2022}, with the main novelty being the support for under-approximating monadic SE. 
It provides a formal, mental model of the \sot implementation, and client implementations of the \myinline{Symex} interface must be correct against this formalisation. 

\paragraph{Values and Symbolic Variables}
Recall from \cref{sec:overview} that SE uses \emph{symbolic variables} to represent multiple values at once. 
\sot is parametric on the values desired for execution, supplied by the user. 
As such, we assume the user provides a set of values, $\vals$, together with a set of \emph{sorts} $\sorts \ni \sort$ with $\sorts \eqdef \mathcal{P}(\vals)$.
%
%
We also assume a countably infinite set of variable identifiers, $\varids$, and define symbolic variables $\sx, \sy \in \svars \eqdef \varids \times \sorts$ (as a convention, we denote symbolic entities with a hat) as pairs of a variable identifier and a sort. 
We write $\sx : \sort$ to denote that $\sx$ is of sort $\sort$.

Symbolic variables and concrete values are connected through \emph{symbolic interpretations}, $\sint \in \sints \eqdef \svars \parmap \vals$, defined as partial maps from symbolic variables to values,  enforcing the following \emph{well-sortedness invariant}: for all $\sx, v$, if $\sx : \sort$ and $\sint(\sx) = \val$, then $\val \in \sort$. 

\paragraph{Symbolic Abstractions}
Given a set $A$, a \emph{symbolic abstraction} over $A$, $\symb{a} \in \symb{A}$,  is an object that is interpreted as a set of elements of $A$ under some $\sint$.
Formally, there must exist a relation $\models~\subseteq \sints \times A \times \symb{A}$ such that $\sint, a \models \symb{a}$ means that $a$ is one of the possible values represented by $\symb{a}$ under interpretation $\sint$.
%

A simple example of symbolic abstraction is \emph{symbolic values}, which are symbolic abstractions over the set of values $\vals$. 
Symbolic values can also be built by lifting operators over values to symbolic values in the natural way.
For instance, the value $\sx + 1$ is the symbolic value that represents the singleton set $\{ \sint(\sx) + 1 \}$ under interpretation $\sint$. Similarly, symbolic maps as presented in \cref{sec:data} are symbolic abstractions of finite partial maps. In general, the \myinline{Data} module of \sot contains implementations for symbolic abstractions over data stuctures.

\paragraph{Symbolic Booleans and Solvers}
To define the SE monad we need the notion of \emph{symbolic booleans}, $\pcs \ni \pc$, which are symbolic abstractions over booleans $\bools$. 
A symbolic boolean $\pc$ is \emph{satisfiable}, denoted $\sat(\pc)$, if there exists an interpretation $\sint$ such that $\sint, \vtrue \models \pc$. 
Conversely, it is \emph{unsatisfiable}, denoted $\unsat(\pc)$, if there exist no such interpretation.

A \emph{solver} is an oracle that can determine the satisfiability (or unsatisfiability) of symbolic booleans. 
In practice, however, solvers are not perfect. 
For instance, SMT solvers are often based on incomplete theories and may return $\unknown$ when queried.
To address this, we interpret $\unknown$ results \emph{approximately}, based on the analysis \emph{mode} $m$, which may be under- ($\mux$) or over-approximate ($\mox$). 
Specifically, SE in \mux must not deem an infeasible path as feasible, and thus it must interpret $\unknown$ as $\unsat$. 
Conversely, SE in \mox must interpret $\unknown$ as $\sat$.
We generalise this by defining $\mode$-\emph{approximate} solvers $\sat_\mode$, where $\mode \in \{ \mox, \mux \}$, such that:\\
\centerline{
$
	\sat(\pc) \Rightarrow \sat_\mox(\pc) 
	\qquad
	\text{and}
	\qquad
	\sat_\mux(\pc) \Rightarrow \sat(\pc)
$
}
\begin{remark}
Another practical consideration that tends to be ignored by existing formalisations of symbolic execution is that of \emph{partial functions}. 
Specifically, SMT-LIB-compatible solvers such as Z3 or CVC5 consider partial functions as \emph{unspecified} when applied outside their domain. 
For instance, Z3 deems the expression $\sy = \sx / 0$ to be \emph{satisfiable} for all $\sx$ and $\sy$, because division by zero makes the expression unspecified, not undefined. 
Therefore, the symbolic boolean $\sx / 0$ must be viewed as a symbolic abstraction that can be interpreted to \emph{any integer} under any interpretation. 
\end{remark}


\paragraph{Branches} A \emph{branch},  $\braket{a}{\pc} \in \braket{A}{\pcs}$, is a pair comprising a result value $a \in A$ and a \emph{path condition} (a symbolic boolean) $\pc \in \pcs$. 
%
We can now define the SE monad as follows:
\[
\begin{array}{@{} r @{\hspace{2pt}} l@{}}
	\symexM(A) \eqdef 
	& \pcs \rightarrow \poset(\braket{\symb{A}}{\pcs}) 
	\hfill
	\mathtt{return}(a) \eqdef \lambda \pc.~\{ \braket{a}{\pc} \} \\
	\mathtt{bind}(m,f) \eqdef
	& \lambda \pc.~ \{ \braket{b}{\pc''} \mid \braket{a}{\pc'} \in m(\pc) \land \braket{b}{\pc'} \in f(a, \pc') \}
\end{array}
\]
%
%
The \myinline{return} function returns a single branch containing the given value and does not modify the path condition. 
The \myinline{bind} function takes each branch of the first computation and feeds its value and path condition to the second computation, collecting all resulting branches.
We can similarly define other operations in the \myinline{Symex} interface; \eg we can define \myinline{nondet}  and \myinline{if\%sat} as follows:
\vspace{-5pt}
%
\begin{align*}
	\mathtt{nondet}(\sort) \eqdef 
	& \lambda \pc.~ \{ \braket{\sx}{\pc} \} \texttt{ where } \sx : \sort \text{ is fresh} \vspace{-5pt}\\
	\mathtt{branch\_on}(\symb{b},t,e) \eqdef 
	& \lambda \pc.~ \big(t(\pc \land \symb{b}) \text{ if } \sat_\mode(\pc \land \symb{b}) \text{ else } \emptyset\big)
 	\cup
 	\big(e(\pc \land \neg \symb{b}) \text{ if } \sat_\mode(\pc \land \neg \symb{b}) \text{ else } \emptyset\big)
\end{align*}
\vspace{-15pt}
%
%

That is, as described in \cref{sec:implem}, symbolically executing a branch adds the guard ($\symb b$) and its negation ($\neg\symb b$) to the path condition and executes the corresponding branch if the resulting path condition is satisfiable; otherwise, it discards the branch. 
If both are satisfiable, it explores both branches.

\paragraph{Soundness}
We define soundness for a \emph{SE computation}, \ie a function $\symb{A} \rightarrow \symexM(\symb{B})$, with respect to a nondeterministic computation, \ie a function $A \rightarrow \poset(B)$.
We define both $\mox$ and $\mux$ soundness (respectively desirable for verification and bug-finding) in the \ifbool{extendedversion}{technical appendix}{extended version~\cite{extended-version}}. 
Let $f : A \rightarrow \poset(B)$ be a nondeterministic function, $\symb{A}$ and $\symb{B}$ be symbolic abstrations over $A$ and $B$, and $\symb{f} : \symb{A} \rightarrow \symexM(\symb{B})$. 
Intuitively, a function $\symb{f}$ is \emph{$\mox$-sound} with respect to $f$, written $\symb{f} \soundwrtox f$, if each outcome of $f$ is covered by a branch of $\symb{f}$ (is also an outcome of $\symb{f}$).
We define $\mux$-soundness analogously; we elide the formal definitions here and present them in the \ifbool{extendedversion}{technical appendix (\cref{app:formalisation})}{extended version~\cite{extended-version}}.

A simple example of a both $\mox$- and $\mux$-sound symbolic computation is \myinline{nondet}, which is sound with respect to the nondeterministic function that returns all values of the given sort. More complex symbolic computations are built by composing simple symbolic computations using $\mathtt{bind}$ and $\mathtt{branch\_on}$. We show that these two operations \emph{preserve} soundness.
Specifically, \cref{thm:bind-branch-sound-app} (proof in \ifbool{extendedversion}{\cref{app:formalisation}}{the extended version~\cite{extended-version}}) states that 
\begin{enumerate*}
	\item composing two $\mode$-sound symbolic computations using $\mathtt{bind}$ yields an $\mode$-sound symbolic computation (this  ensures the soundness of sequentially composing computations); and
	\item $\mathtt{branch\_on}$ is a symbolic lifting of the concrete \myinline{if...else...}~construct.
\end{enumerate*}

\begin{theorem}[Soundness Preservation]\label{thm:bind-branch-sound-app}
For all $f, \symb{f}, g, \symb{g}, a, \symb{a}, b, \symb{b}$ and all modes $\mode \in \{ \mox, \mux \}$:
\begin{enumerate}
	\item if $\symb{f} \soundwrtm f$ and $\symb{g} \soundwrtm g$, 
then $\symb{g} \fishop \symb{f} \soundwrtm g \fishop f$, where $p \fishop q \eqdef \lambda x.~ \mathtt{bind}~p(x)~q$, denoting the Kleisli composition (`fish') operator.
	\item Let $h \eqdef \mathtt{if}~b~\mathtt{then}~f~a~\mathtt{else}~g~a$
	and $\symb{h}~(\symb{b}, \symb{a}) \eqdef \mathtt{branch\_on}~\symb{b}~(\symb{f}~\symb{a})~(\symb{g}~\symb{a})$.
	If $\mathtt{branch\_on}$ uses an $\mode$-approximate solver, 
	$\symb{f} \soundwrtm f$ and $\symb{g} \soundwrtm g$, 
	then $\symb{h} \soundwrtm h$.
\end{enumerate}
\end{theorem}
\section{Discussion, Limitations and Future Work}
\label{sec:discussions}
We begin by discussing the advantages of two key design decisions we made in \sot, namely \begin{enumerate*}
	\item implementing \sot as a shallowly-embedded library in \ocaml; and 
	\item implementing ``symbolic execution'' (SE), rather than ``symbolic compilation'' that performs path merging and generates a single formula.
\end{enumerate*} 
We then proceed with discussing the limitations of \sot, which naturally lead to possible avenues of future work. 

\subsection{Key Design Decisions in \sot}\vspace{-5pt}
\myparagraph{\sot in \ocaml} Unlike most reusable SE, which provide a single intermediate language (IL) to analyse many languages, \sot is a implemented in \ocaml, providing base abstractions (the SE monad) and reusable components (symbolic computations and data structures) to build custom SE engines for different languages. While we have demonstrated that this method is effective (at least for Rust and C), we point out one additional advantage of this approach.

Specifically, by forgoing an IL, we can leverage all features of our host language (\ocaml) when implementing engines. For instance, in our implementation of \lstinline{BV_value}, we make extensive use of \ocaml's ability to attach \emph{ghost types} to values, allowing us to ensure \eg that ill-typed expression cannot be constructed. Doing so in an IL would require implementing a custom type system from scratch, requiring substantially more work and likely leading to a reduced expressiveness in comparison to \ocaml's mature type system.
Another example is our use of \ocaml's module system which allows us to use functors, thereby seamlessly enables us to make an engine parametric on various choices, including the choice of the symbolic state representation, or the symbolic pointer representation. Using an IL, one would have to embed this parametricity in the IL itself.

\myparagraph{Symbolic Execution versus Symbolic Compilation for Bug Finding} \sot performs symbolic \emph{execution}: it explores all feasible paths of the program separately and generates a formula for each path. An alternative approach is to perform symbolic \emph{compilation}, where the engine performs path merging and generates a single formula for the whole program.
Nevertheless, we believe it is possible in \sot to implement path merging on a given object of type \myinline{'a Symex.t} (a computation that yields a value of type \myinline{'a}), as long as a function \myinline{merge: Value.t -> 'a -> 'a -> 'a} is provided; this is in line with the approach taken in Grisette~\cite{grisette-symbolic-compilation-lu-2023}.

However, existing data suggests that while symbolic compilation shines when merging path conditions over simple values such as bitvectors (as in the \myinline{ctpop} example of \cref{fig:ctpop-tests}), it does not perform as well when merging path conditions over expressions that characterise more complex data structures, \eg objects in the heap, which can overwhelm the solver. 
An example of the latter is inserting $n$ non-deterministic integers into a Rust \myinline{BTreeSet} and checking that the resulting set is ordered. 
For $n=3$, \rusteria explores 13 paths in 2.09s, while Kani merges all paths into a single SAT instance of
${\sim}50$M variables, taking 691s to solve; for $n=6$, \rusteria
explores 4{,}683 paths in 456s, while Kani exceeds a two-hour timeout (see \ifbool{extendedversion}{\cref{app:btreemap}}{the extended version~\cite{extended-version} for more details}).

Moreover, we envision \sot being used primarily to implement \emph{bug detection} (rather than \emph{verification}) tools, where path merging may lead to imprecise abstractions~\cite{limits-difficulties-design-ascari-2024}. In the bug detection setting we believe that path exploration is more effective,  and it is the approach taken by the industrial bug detection engine Infer.Pulse~\cite{incorrectness-logic-ohearn-2019,local-reasoning-presence-raad-2020}. 

\subsection{Limitations and Future Work}
\label{sec:limitations}
%

\myparagraph{Semantic Coverage} \rusteria and \coteria do not currently support various features such as concurrency (through the async/await constructs),  FFI calls, or inline assembly.
In the future, we will extend the coverage of both interpreters by exploring foreign calls between C and Rust and leveraging both interpreters. 
We will further explore how to leverage the guarantees offered by the semantics of Rust to support bug finding in concurrent code that uses async/await constructs.

A main roadblock in extending our coverage thus far lies with the \emph{frontends} we use. 
For instance, Cerberus (our \coteria frontend) does not currently support several built-in atomic operations used by widely-used libraries such as \lstinline|libgit2| . These operations are not part of the C11 standard; they are compiler-specific extensions provided by Clang and GCC. 
Lack of support for these operations limits the applicability of \coteria, as it prevents extracting an AST for C programs using them. 


\myparagraph{Solvers} While we provide blueprints and a clear signature for custom new values and solvers, we believe a substantial part of the required work  could be \emph{automated} through meta-programming. 
In the future, we will create a DSL to specify new symbolic values, their operations, and simplifications that can be extracted to interactive theorem provers to be proven sound. 
We will further provide reusable types with libraries of simplifications for common symbolic values such as bitvectors.
Finally, we will explore adjusting the settings and default tactics of Z3 to improve performance on our generated queries,  and we will explore other solvers as well.

\myparagraph{Polymorphic and Compositional Reasoning for Rust} We will extend \rusteria to support executing polymorphic functions, leveraging its parametric design and its Charon frontend.
This will allow executing symbolic tests that model execution for all type instantiations of a polymorphic function.
We will then explore compositional bug finding for Rust by leveraging this polymorphic support and the existing bi-abduction support in \sot. 

\myparagraph{Lean} Our current \sot implementation benefits from our extensive knowledge of \ocaml for efficiency and ease of maintainability.
However, since \sot is a functional library, it can be implemented in any functional language with support for monads. In the future, we will explore re-implementing \sot in Lean. We will explore other, better-suited optimisations, leverage Lean's powerful meta-programming capabilities, and investigate the interaction between symbolic execution and interactive theorem proving in this setting.
\section{Related work}
\label{sec:relwork}

\myparagraph{Monadic Symbolic Execution}
Previous work has explored using monads to structure SE engines, in more theoretical settings.
\citewithauthor{abstracting-definitional-interpreters-darais-2017} propose a methodology to lift a definitional interpreter for a given language written in a monadic style into various abstract interpreters, including a symbolic one, but do not explore an implementation for a real programming language, or how to optimise it.
\citewithauthor{verified-symbolic-execution-keuchel-2022}, on the other hand, present a monadic approach mechanised in Rocq and implement in \textsc{Katamaran}, a tool for proving ISA security properties.
While \textsc{Katamaran} comes with a mechanised proof of soundness, its main aim is to improve the proof engineering experience inside Rocq.
They do this by using a deep embedding of path conditions that enables manual simplifications without resorting to meta-programming, and their monadic approach is \emph{not executable}.
\sot, in contrast, explores an executable and \emph{reusable} approach to monadic SE.

Owi~\cite{owi-performant-parallel-andres-2024} is an SE tool for Wasm that uses a similar monadic structure to \sot.
It is designed for performant, parallel exploration of multiple paths -- a capability we have not yet exploited in \sot.
Unlike \sot, Owi treats Wasm as a fixed IL for analysing programs compiled from C, Rust, Zig, and so forth. 
As such, unlike \rusteria, Owi does not support Rust-specific features such as \TBs, because the resulting Wasm code lacks the necessary information.

\myparagraph{Abstract Interpretation for Static Analysis} There are many scalable static analysers based on abstract interpretation~\cite{abstract-interpretation-unified-cousot-1977,framac-software-analysis-kirchner-2015,why-does-astree-cousot-2009}. In particular, MOPSA~\cite{design-modular-platform-static-mine-2018,easing-maintenance-academic-monat-2024} is a multi-language analysis platform (currently for C and Python) that (as with \sot) aims for modularity and reusability. For instance, it uses \ocaml features such as extensible variants to enable elegant reuse of transfer functions across languages. Currently, MOPSA is a whole-program (non-compositional) static analyser and targets OX analyses, whereas \sot primarily focuses on UX bug finding.

\myparagraph{Semantic Framework and Symbolic Lifting} 
Rosette~\cite{rosette-growing-solver-aided-torlak-2013,lightweight-symbolic-virtual-torlak-2014} extends Racket \cite{racket} with solver-aided programming features, automatically lifting a concrete interpreter for a language written in Racket into a symbolic one using symbolic compilation (that is, compiling the entire program to a single SMT formula). 
Grisette~\cite{grisette-symbolic-compilation-lu-2023} later formulated this approach, as well as various optimisations, using a monadic approach, providing a functional programming library in Haskell. In principle, one could encode analyses such as ours in Rosette and Grisette, though previous experiences have shown that this approach struggles to scale to real world code~\cite{symbolic-execution-javascript-santos-2018,javert-20-compositional-fragososantos-2019}.

Similarly, the $\mathbb{K}$ framework~\cite{overview-semantic-framework-rosu-2010} allows one to define the formal semantics of a language using rewriting logic, and it has an SE backend for performing symbolic testing. 
However, defining semantics in $\mathbb{K}$ entails substantial work:
the KMIR project~\cite{introducing-kmir-concrete-cumming-}, commenced in 2023, aims to provide a K semantics and SE for Rust; however, it does \emph{not} support \TBs and does not seem to be ready for use -- we could not find any existing symbolic tests.

\myparagraph{Rust Analysis Engines} There are several tools for semi-automatic verification of Rust programs~\cite{leveraging-rust-types-astrauskas-2019,creusot-foundry-deductive-denis-2022,refinedrust-type-system-gaher-,verus-verifying-rust-lattuada-2023,hybrid-approach-semiautomated-ayoun-2025,flux-liquid-types-lehmann-2022,aeneas-rust-verification-ho-2022,modular-formal-verification-foroushaani-2022}. 
While they all provide more guarantees than \rusteria, they all require \emph{manual user annotations} to specify pre- and post-conditions, loop invariants, and at times even tactics to guide the proof. 
Moreover, \emph{none} support reasoning about \TBs. 
Kani~\cite{kani-verifying-dynamic-trait-vanhattum-2022} is an industry-grade, state-of-the-art symbolic testing tool for Rust, and it does so by compiling Rust code to the CBMC IL~\cite{cbmc-bounded-model-kroening-2014}. 
We compared \rusteria and Kani at length in \cref{sec:soteria-rust}.

\myparagraph{Bi-abduction and Bug Finding} To our knowledge, there are two tools that perform bi-abductive SE for automatic bug finding: Infer~\cite{infer-automatic-program-calcagno-2011,finding-real-bugs-le-2022} and Gillian~\cite{gillian-part-multilanguage-fragososantos-2020,gillian-part-ii-maksimovic-2021,compositional-symbolic-execution-loow-2024}. 
\infertool is an industrial-strength tool by Meta that supports multiple languages. 
However, \infertool is based on a single IL, SIL, leading to complex compilers. 
As such, extending \infertool to a new language often requires modifying \infertool itself by writing `models' for some functions directly in \ocaml.
Gillian is the framework that is  closest to our goal of adapting the tool to each language. 
Gillian allows one to override the notion of memory model directly in \ocaml for each source language. 
However, Gillian still relies on an GIL, its IL, and each language instantiation must provide both a complex compiler to its IL and an  \ocaml implementation of all state operations. 
In addition, Gillian's expression language is fixed and captures constructs that must be compatible with both JavaScript and C, leading to a complex expression language that still requires constant decoding and re-encoding of values (which is our leading hypothesis for its performance compared to \sot).

\ifbool{reviewversion}{}{
\section*{Data-Availability Statement}
We provide the full implementation of \sot described in this paper, as well as all the benchmarks and scripts used to reproduce the results in this paper in the accompanying artifact, available at \url{https://doi.org/10.5281/zenodo.19080424}~\cite{artifact}.}

\begin{acks}
This project was sponsored by the Defense Advanced Research Projects Agency, (DARPA),
Information Innovation Office (I2O), Program BAA HR001124S0003, under Cooperative Agreement
No. HR00112420359.  Disclaimer: The content of the information does not necessarily reflect the
position or the policy of the U.S. Government, and no official endorsement should be inferred. 

Raad is further supported by the UKRI Future Leaders Fellowship MR/V024299/1, the EPSRC grant
EP/X037029/1 and VeTSS.
We would like to thank Guillaume Boisseau for his help with making Charon work for \rusteria; 
Petar Maksimovi\'{c} for his feedback on the earlier drafts of this manuscript; 
the Cerberus developers for their help with setting up, using and improving Cerberus; 
Peter Sewell for hosting Ayoun as a visiting researcher at the University of Cambridge; 
and Peter O'Hearn for his guidance throughout this project.
Finally, we thank the PLDI 2026 reviewers for their feedback and suggestions,
which have greatly improved the quality of this paper.
\end{acks}

\bibliographystyle{ACM-Reference-Format}
\bibliography{soteria}

\ifbool{extendedversion}{
\appendix

\clearpage
\section{The Simple \simlang Language}
\label{app:simple-lang}

We now present the simple programming language \simlang used throughout this paper to illustrate our concepts.
\simlang is a ML-like, expression-based language with a small set of constructs, to illustrate a full symbolic execution engine.

We first present the abstract syntax and concrete semantics of \simlang in \cref{sec:simple-lang-semantics}. We then present how to implement a symbolic execution engine for \simlang using \sot in \cref{sec:symex-simple-lang}, while giving a brief introduction to symbolic execution and monads along the way.

\subsection{\simlang Semantics and Interpreter}
\label{sec:simple-lang-semantics}

For uniformity, we define the abstract syntax of \simlang using \ocaml notation in \cref{fig:lang} (above).
\simlang has two kinds of constants (integers, \myinline{Z.t}, being the type of unbounded integers,  and booleans) a few binary operators (addition, division, equality and boolean conjunction), and six kinds of expressions: constants, variables, binary operators, let-binding, if-then-else expressions, a non-deterministic integer expression that can evaluate to any integer, and an error expression that terminates evaluation in error mode if its argument evaluates to false.

\begin{figure}[!t]
\hfill
\begin{lstlisting}
type const = Int of Z.t  | Bool of bool
type binop = Add | Sub | And | Div | Eq | Geq
type expr = Const of const | Var of string | BinOp of binop * expr * expr
            | Let of string * expr * expr | If of expr * expr * expr
            | NondetInt | Assert of expr
\end{lstlisting}

\begin{mathpar}
\inferrule[If-Then]{
  \semtrans{\subst}{g}{\ook: \vtrue}\quad
  \semtrans{\subst}{\expr_1}{\outcome: \val}
}{
  \semtrans{\subst}{\ite{g}{\expr_1}{\expr_2}}{\outcome: \val}
}%
\and%
\inferrule[If-Else]{
  \semtrans{\subst}{g}{\ook: \vfalse}\quad
  \semtrans{\subst}{\expr_2}{\outcome: \val}
}{
  \semtrans{\subst}{\ite{g}{\expr_1}{\expr_2}}{\outcome: \val}
}%
\and%
\inferrule[Let-Success]{
  \semtrans{\subst}{e1}{\ook: \val}\quad
  \subst' = \subst\mapupdate{\px}{\val}\\\\
  \semtrans{\subst'}{e2}{\outcome: \val'}
}{
  \semtrans{\subst}{\letin{\px}{\expr_1}{\expr_2}}{\outcome: \val'}
}%
\and%
\inferrule[Let-Error]{
  \semtrans{\subst}{\expr_1}{\oerr: \val}
}{
  \semtrans{\subst}{\letin{\px}{\expr_1}{\expr_2}}{\oerr: \val}
}%
\and%
\inferrule[NondetInt]{
  z \in \ints
}{
  \semtrans{\subst}{\nondetint}{\ook: z}
}
\end{mathpar}
\vspace{-10pt}
\caption{The \simlang abstract syntax tree (above); an excerpt of its big-step operational semantics (below)}
\label{fig:lang}
\label{fig:simple-conc-sem}%
\end{figure}

We present an excerpt of the straightforward big-step semantics of \simlang in \cref{fig:simple-conc-sem} (below).
A semantic judgement $\semtrans{\subst}{\expr}{\outcome:v}$ states that evaluating expression $\expr$ under substitution $\subst$ yields value $v$ with outcome $\outcome$, where $\subst : $ \myinline{subst = (var, value)  map} is a partial map from program variables to their values, $o :$ \myinline{(value, error) result} indicates terminating either successfully returning a value denoted by $\ook~\val$, or  erroneously carrying an error denoted by $\oerr~\err$, where \myinline{value} $\approx$ \myinline{const}.

These five rules highlight central concepts to defining most program semantics: \emph{conditional choice}, \emph{sequential composition}, \emph{error propagation} and \emph{non-determinism}.
Conditional choice is the ability to choose between execution paths based on a condition.
Specifically, the \inferref{If-Then} and \inferref{If-Else} rules apply when the condition evaluates to \myinline{true} and \myinline{false}, respectively.
Sequential composition is captured by the \inferref{Let-Success} rule,  requiring sub-expression $\expr_1$ to be evaluated \emph{before} sub-expression $\expr_2$, as the latter is evaluated in the context of a substitution that depends on the result of the former.
\inferref{Let-Error} showcases a simple model of error propagation: an erroneous evaluation is propagated and terminates the entire program (`short-circuits').
Finally, \inferref{NondetInt} models non-determinism: evaluating $\nondetint$ may yield \emph{any} integer; \ie there is an infinite number of possible results.
While real programming language rarely have such a construct, non-determinism can be used to model \eg indeterminate user input, or whether \myinline{malloc} should return a \nullptr pointer on allocation.

\subsection{Symbolic Execution Monad for \simlang}
\label{sec:symex-simple-lang}
Before defining our \emph{symbolic} execution monad, we present the intuition behind the \emph{concrete} execution monad underpinning expression evaluation rules in \cref{fig:simple-conc-sem}.
We then present the \ocaml feature that enables defining custom let-binding operators, allowing us to write the rules above \emph{succinctly}.

\subsubsection{The Problem}
The semantics in \cref{fig:simple-conc-sem} is defined as a relation using inference rules.
While common and useful when proving properties about the semantics, it is less readable to the common developer and less straightforward to implement.
As our goal is to implement a symbolic interpreter, we should ideally re-define the semantics as a \emph{function} as in the following \ocaml snippet:
\begin{lstlisting}
let rec eval (subst: subst) (e: expr) : (Value.t, error) Result.t =
  match e with
  | Let (x, e1, e2) ->
    let v1 = eval subst e1 in
    let subst' = Subst.set subst x v1 in
    eval subst' e2
  | ... -> ...
\end{lstlisting}

Unfortunately, while this reads naturally, it does not model the semantics in \cref{fig:simple-conc-sem} accurately: it does not capture that evaluating \myinline{e1} may fail or yield multiple values. Cue the \emph{execution monad}.

\subsubsection{Execution Monad}
The semantics judgement in \cref{fig:simple-conc-sem} can also be defined as a function: \myinline{eval : subst -> expr -> (value, error) result set}, mapping an expression and substitution a \emph{set} of possible outcomes.
This set is, in general, non-computable; \eg evaluating $\nondetint$ under any substitution returns the infinite set $\left\{\ook~z \mid z \in \ints \right\}$.
The return type of the above \myinline{eval} function is a monad, which for the sake of this presentation we name the \emph{execution monad}:%
\footnote{This execution monad is trivially obtained by applying the result monad transformer to the non-determinism (set) monad,y}
\begin{lstlisting}
type 'a exec = ('a, error) result set
\end{lstlisting}
%
We proceed with an intuitive description of what a monad is. 

\subsubsection{An Introduction to Monads}
A monad is a parametrised type \myinline{'a t} with two operations:
\begin{lstlisting}
val return : 'a -> 'a t
val bind : 'a t -> ('a -> 'b t) -> 'b t
\end{lstlisting}
where \myinline{return} characterises a \emph{pure} operation that lifts any value to the monad, and \myinline{bind x f} receives \myinline{x} (the result of a monadic computation) and \myinline{f} (a monadic computation) and \emph{composes} the two.
These two operations must satisfy \emph{monad laws}, stating natural properties such as computation associativity.
We elide these laws as there is extensive descriptions of monads in the literature.
For the monad presented in this section, we prove in \cref{app:formalisation} that the monad laws hold.

For our execution monad, \myinline{return} is the function that receives a value $v$ and returns the singleton set containing $\ook~v$.
It is `pure' in that no effect -- non-determinism or error -- is performed.
The \myinline{bind} operation receives a set of results and applies \myinline{f} to each successful result, collecting their results.
If any of the branches yields an error, the error is simply propagated:

\vspace{-1em}
\[
\mathtt{bind}~\mathtt{x}~\mathtt{f} = \eqnmark{node1}{\left\{ \outcome' \mid \exists \val.~ \ook~\val \in \mathtt{x} \land \outcome' \in \mathtt{f}~\val \right\}} \cup \eqnmark{node2}{\left\{ \oerr~\err \mid \oerr~\err \in \mathtt{x} \right\}}
\]
\setannotstyle{color=softgreen}
\annotate{below,left}{node1}{$\mathtt{f}$ is applied to each successful result}
\setannotstyle{color=red}
\annotate{below,right}{node2}{errors are propagated}
\vspace{0.3em}

\noindent With this monadic structure, we can re-write the rules from \cref{fig:simple-conc-sem} as follows:

\begin{lstlisting}
let rec eval subst expr =
  match expr with
  | Let (x, e1, e2) ->
    bind
      (eval subst e1)
      (fun v ->
        let subst' = Subst.set subst x v in
        eval subst' e2)
  | ... -> ...
\end{lstlisting}

\subsubsection{Custom Let Operators in \ocaml}
While the above is more readable to an average (\ocaml) developer,  it is still not very readable as calls to \myinline{bind} create visual noise.
\ocaml, inspired by Haskell's do-notations, provides a syntactic sugar for monadic binding:
\begin{lstlisting}[escapeinside={&}{&}]
&\tikzmark{let-bind-def-start}&let ( let* ) = bind&\tikzmark{let-bind-def-end}&
let rec eval subst expr =
  match expr with
  | Let (x, e1, e2) ->
    &\tikzmark{let-star-start}&let*&\tikzmark{let-star-end}& v1 = eval subst e1 in
    let subst' = Subst.set subst x v1 in
    eval subst' e2
  | ... -> ...
\end{lstlisting}
\tikz[overlay,remember picture]{
  \draw[fill=softgreen!80,draw=none,opacity=.3]
    ($(pic cs:let-bind-def-start) + (-.2em,.9em)$) rectangle ($(pic cs:let-bind-def-end) + (.2em,-.4em)$);
}%
\tikz[overlay,remember picture]{
  \draw[fill=softgreen!80,draw=none,opacity=.3]
    ($(pic cs:let-star-start) + (-.2em,.9em)$) rectangle ($(pic cs:let-star-end) + (.2em,-.4em)$);
}%

This syntactic sugar allows us to write the code in a sequential style, elegantly hiding the monadic structure and making it natural to read, as shown above.
Now that we can write the semantics of \simlang as a function, we next show how \sot can help us write a symbolic interpreter.

\subsubsection{A Primer on Symbolic Execution}\label{sec:symex-intro}
Symbolic execution is a program analysis technique that specifically addresses the problem of non-determinism in testing.
It enables testing of programs that either contain natural non-determinism (such as allocation in C which may non-deterministically%
\footnote{Allocation is not non-deterministic and only fails when the OS is unable to provide the required memory; however,  from the point of view of exhaustive testing, it can be considered non-deterministic.}
fail to allocate), or artificial non-determinism introduced to extend input coverage.

For instance, CBMC, KLEE, Gillian-C or Kani all enhance their target language with a function that creates a non-deterministic value, similar to our $\nondetint$ construct in \simlang.
This construct can be used to test properties across all possible inputs of a program with the following simple test:
\begin{lstlisting}
let y = NondetInt in
let v = if y < 0 then 0 - y else y in
assert (v >= 0)
\end{lstlisting}

Unfortunately, the above test is not executable as it requires exploring an infinite number of cases.
Symbolic execution addresses this by abstracting over non-determinism through the use of \emph{symbolic values}, \ie values that depend on \emph{symbolic variables}.
Throughout the execution, a symbolic interpreter uses an SMT solver~\cite{z3-efficient-smt-demoura-2008, cvc5-versatile-industrialstrength-barbosa-2022} to prune out infeasible paths.

For instance, the expression presented in the above simple test starts from assigning a non-deterministic integer to \myinline{y}, binds its absolute value to \myinline{v}, and asserts that \myinline{v} is non-negative.
A symbolic execution of this expression starts by assigning a symbolic variable $\symb{y}$ to program variable \myinline{y}.
The evaluation of \myinline{if} expression starts by querying an SMT solver to check whether the condition \myinline{y < 0} (which evaluates to the symbolic value $\symb{y} < 0$) is satisfiable.
Since it is,  \ie there exist an integer that is strictly negative, the symbolic execution engine evaluates the `then' \emph{branch}.
In addition, since the negation of the condition ($\symb{y} \geq 0$) is also satisfiable, the symbolic execution engine also evaluates the `else' branch.
In each execution branch, the symbolic execution engine maintains a \emph{path condition} recording the conditions under which the branch is executed.
Subsequently, the symbolic execution engine evaluates the \myinline{assert (v >= 0)} expression in each execution branch, ensuring that \myinline{v >= 0} holds by checking that its negation (\myinline{v < 0}) is unsatisfiable.
Since it is unsatisfiable in both branches, the assertion is verified \emph{for all possible inputs}.
\sot allows us to implement a symbolic execution engine that performs exactly this kind of execution, as we describe below.

\subsubsection{Symbolic Values and Solvers}

The first step to using \sot is to instantiate its symbolic execution monad with a \myinline{Solver}.
A \myinline{Solver} is a module that provides
\begin{enumerate*}
	\item a \emph{symbolic} value language, and
	\item a solver interface to check satisfiability for the subset of these symbolic values that are boolean expressions.
\end{enumerate*}
This parametricity allows \sot to be used with any abstraction of symbolic values (\ie the set of values over which symbolic variables range) and with any encoding to a solver (\eg an off-the-shelf SMT solver such as Z3, a custom solver, or any combination thereof).

In our experience using other tools such as Gillian~\cite{gillian-part-multilanguage-fragososantos-2020,gillian-part-ii-maksimovic-2021}, trying to support multiple languages within a single set of symbolic values (through a single intermediate language) leads to complex languages that are the source of bugs and performance issues in the symbolic execution engine and make the intermediate language an error-prone target for front-ends.
However, two languages that operate on a similar level of abstraction, such as C and Rust, can easily share a symbolic value language.

For this presentation we re-use \myinline{Bv_solver}, one of our pre-built solvers that come with the \sot base library.
\myinline{Bv_solver} provides a simple symbolic value language that is suitable for C-like languages, using an efficient incremental interface, built-in simplifications with Z3 as its underlying SMT solver.
Symbolic variables range over integers, booleans, bit-vectors and pointers (a pair of integers denoting an allocation identifier and offset).

\begin{lstlisting}
module Symex = Soteria.Symex.Make(Soteria.Bv_solver.Bv_values)
\end{lstlisting}\vspace{-.5em}

\subsubsection{Symbolic Interpreter}
The \myinline{Symex} module exposes a symbolic execution monad over the values of \myinline{Bv_values}, as well as all primitives required to implement a symbolic interpreter.
In particular, while the \myinline{Symex} monad is an abstraction over the non-determinism monad presented above, the \myinline{Symex.Result} module provides an abstraction over the execution monad; that is, the monad that captures both non-determinism and error propagation.

\Cref{fig:comparison-eval} presents the (monadic) \emph{concrete} evaluation of the \myinline{If} and \myinline{Let} constructs of \simlang implemented using \sot (left) alongside its \emph{symbolic} counterpart (right).
Notably, the two definitions are almost identical, except that (as highlighted):
\begin{enumerate*}
  \item the concrete interpreter uses the bind operation for \myinline{ExecutionMonad} while the symbolic interpreter uses \myinline{Symex.Result}, which performs symbolic execution and stops each execution path upon error;
  \item the concrete interpreter uses the native \myinline{if} of \ocaml to branch on a concrete boolean, while the symbolic interpreter uses \sot's \myinline{if\%sat}, which branches on a symbolic boolean by consulting the solver and explores a branch whenever its path condition is satisfiable (both branches may be taken if both conditions are satisfiable).
\end{enumerate*}

\begin{figure}%
\begin{minipage}{0.48\textwidth}
\begin{lstlisting}[escapeinside={&}{&}]
&\tikzmark{conc-let-start}&let ( let* ) = ExecutionMonad.bind&\tikzmark{conc-let-end}&
...
| If (cond, then_e, else_e) ->
  let* cond_v = eval subst cond in
  if cond_v
  then eval subst then_e
  else eval subst else_e
| Let (x, e1, e2) ->
  let* v1 = eval subst e1 in
  let subst' =
    Subst.set subst x v1 in
  eval subst' e2
\end{lstlisting}%
\end{minipage}%
\hfill%
\begin{minipage}{0.48\textwidth}
\begin{lstlisting}[escapeinside={&}{&}]
&\tikzmark{sym-let-start}&let ( let* ) = Symex.Result.bind&\tikzmark{sym-let-end}&
...
| If (cond, then_e, else_e) ->
  let* cond_v = eval subst cond in
  if&\tikzmark{sat-start}&%sat&\tikzmark{sat-end}& cond_v
  then eval subst then_e
  else eval subst else_e
| Let (x, e1, e2) ->
  let* v1 = eval subst e1 in
  let subst' =
    Subst.set subst x v1 in
  eval subst' e2
\end{lstlisting}%
\end{minipage}%
\tikz[overlay,remember picture]{
  \draw[fill=red!80,draw=none,opacity=.3]
    ($(pic cs:conc-let-start) + (-.2em,.9em)$) rectangle ($(pic cs:conc-let-end) + (.2em,-.4em)$);
}%
\tikz[overlay,remember picture]{
  \draw[fill=softgreen!80,draw=none,opacity=.3]
    ($(pic cs:sym-let-start) + (-.2em,.9em)$) rectangle ($(pic cs:sym-let-end) + (.2em,-.4em)$);
}%
\tikz[overlay,remember picture]{
  \draw[fill=softgreen!80,draw=none,opacity=.3]
    ($(pic cs:sat-start) + (0em,.9em)$) rectangle ($(pic cs:sat-end) + (0em,-.4em)$);
}%
\vspace{-1em}%
\caption{Concrete (left) and symbolic (right) evaluation}\label{fig:comparison-eval}%
\end{figure}

In \cref{fig:expr-eval}, we present the full symbolic expression evaluation function for \simlang.
The function \myinline{eval} takes as input a substitution mapping program variables to symbolic values, and an expression to evaluate, and returns a symbolic computation yielding either a symbolic value or an error.
We assume the existence of a module \myinline{Value} representing symbolic values, with standard operations. \myinline{Value.Infix} provides infix operators for symbolic values, such as \myinline{( + ) : Value.t -> Value.t -> Value.t}, which are used in \myinline{eval_binop} to symbolically evaluate binary operations while having convenient, concrete-like syntax.

\begin{figure}[!t]
\begin{lstlisting}
module Subst = Map.Make(String)
type subst = Value.t Subst.t

type error = DivisionByZero | AssertError

let value_of_const (c: const) : Value.t =
  match c with
  | Int z -> Value.int z
  | Bool b -> Value.bool b

let eval_binop
  (op: binop) (v1: Value.t) (v2: Value.t) : (Value.t, error) Result.t Symex.t =
  let open Value.Infix in
  match op with
  | Add -> ok (v1 + v2)
  | And -> ok (v1 && v2)
  | Div ->
    if%sat v2 == 0s then error DivisionByZero
    else ok (v1 / v2)
  | Eq -> ok (v1 == v2)
  | Geq -> ok (v1 >= v2)

let rec eval
  (subst: subst) (e: expr) : (Value.t, error) Result.t Symex.t =
  match e with
  | Const c -> ok (value_of_const c)
  | Var x -> ok (Subst.get subst x)
  | BinOp (op, e1, e2) ->
    let* v1 = eval subst e1 in
    let* v2 = eval subst e2 in
    eval_binop op v1 v2
  | Let (x, e1, e2) ->
    let* v1 = eval subst e1 in
    let subst' = Subst.set subst x v1 in
    eval subst' e2
  | If (cond, then_e, else_e) ->
    let* cond_v = eval subst cond in
    if%sat cond_v then eval subst then_e
    else eval subst else_e
  | NondetInt -> nondet Int
  | Assert e ->
    let* v = eval subst e in
    if%sat v then ok 0s
    else error AssertError
\end{lstlisting}
\caption{Full symbolic evaluation function}\label{fig:expr-eval}
\end{figure}

\subsubsection{Guarantees}

The \myinline{Symex} monad provides soundness guarantees with respect to the non-determinism monad, which we informally describe here and prove formally in \cref{app:formalisation}.
For instance:
\begin{itemize}
  \item If two symbolic computations \myinline{fs} and \myinline{gs} are respectively sound against two non-deterministic computations \myinline{fc} and \myinline{gc}, then the composition of \myinline{fs} and \myinline{gs} within the symbolic monad (using the bind operator) is sound against the composition of \myinline{fc} and \myinline{gc} within the non-determinism monad;
  \item The function that branches on a symbolic boolean value using \myinline{if\%sat} is sound against the function that branches on a concrete condition using standard \myinline{if}.
\end{itemize}
By ensuring that each primitive of the \myinline{Symex} monad is sound against its concrete counterpart, and that these primitives compose soundly, \sot helps users build sound interpreters. Users still need to ensure that they indeed compose only sound computations, but this is no more difficult, in our opinion, than writing a correct compiler.


\clearpage
\section{\sot: Supplementary Implementation Details}
\label{app:imp}

\subsection{Detailed implementation of the branching construct}\label{app:branching-impl}

We now provide a detailed description of the implementation of the \myinline{branch_on} operation, and showcase some optimisations and integration with other \sot utilities.

\begin{figure}[!h]
\begin{lstlisting}[numbers=left,xleftmargin=2em,escapechar=^]
(* `if%sat guard then then_ else else_`
    is syntactic sugar for `branch_on guard then_ else_` *)
let branch_on guard then_ else_ = fun f ->
  let guard = Solver.simplify guard in ^\label{line:simplify-guard}^
  match Value.as_bool guard with^\label{line:as-bool-start}^
  | Some true -> then_ () f
  | Some false -> else_ () f^\label{line:as-bool-end}^
  | None ->
    let left_unsat = ref false in
    Symex_state.save ();^\label{line:sat-save}^
    L.with_section "left branch" (fun () ->
      Solver.add_constraints [ guard ];^\label{line:add-guard}^
      let sat_res = Solver.sat () in^\label{line:left-guard}^
      left_unsat := Solver_result.is_unsat sat_res;^\label{line:save-if-unsat}^
      if Solver_result.is_sat sat_res then then_ () f^\label{line:exec-left}^);
    Symex_state.backtrack_n 1;
    L.with_section "right branch" (fun () ->
      Solver.add_constraints [ Value.(not guard) ];
      if !left_unsat then else_ () f
      else
        match Fuel.consume_branching 1 with
        | Exhausted -> Stats.As_ctx.add_unexplored_branches 1
        | Not_exhausted ->
          Stats.As_ctx.add_branches 1;
          if Solver_result.is_sat (Solver.sat ()) then else_ () f)
\end{lstlisting}\vspace{-1em}
\caption{Implementation of the \myinline{if\%sat} (syntactic sugar for \lstinline|branch_on|) construct.}
\label{fig:branch_on}
\label{fig:if_sat}
\end{figure}

The \myinline{branch_on} function receives a guard and two thunks, \myinline{then_} and \myinline{else_}, both returning a symbolic execution monad, as shown in \cref{fig:branch_on}.
We first simplify the guard on line~\ref{line:simplify-guard} using the solver's simplification procedure.
The solving procedure of \myinline{Solver} should be simple and efficient and need not be complete.
In fact, implementing it as the identity function is a valid implementation.
In \cref{sec:implem} we detail the simplification procedure of \myinline{BV_Solver} offered by \sot.
If the solver can simplify the guard to \myinline{true} or \myinline{false}, we can immediately execute the corresponding branch (lines~\ref{line:as-bool-start}--\ref{line:as-bool-end}).

Solvers can often simplify constraints, \eg if the guard or its negation is already part of the current path condition, which happens very often in practice.
This simple optimisation can save a lot of time, as it avoids unnecessary calls to the solver.
In the Collections-C case study presented in \cref{sec:evaluation-c}, out of 67,323 calls to \myinline{branch_on}: 48,112 (71\%) guards were already booleans (due to earlier reductions, or because they are concrete), 6,747 (10\%) guards were simplified to a concrete boolean value by the solver, and 9,002 (13\%) guard terms or their negation were found in the path condition. Overall, only 17,442 $\sat$ checks were performed, which translates into a significant speedup.

If the guard cannot be simplified to a concrete boolean value, we proceed with the symbolic execution of both branches.
To do so, we first save the current path condition (line~\ref{line:sat-save}).
We then execute the left branch by adding the guard to the path condition (line~\ref{line:add-guard}) and calling \myinline{Solver.sat} to check if the current path condition is satisfiable.
Importantly, we record a boolean capturing if the branch is $\sat$ or $\unsat$ (it could also be $\unknown$).
If the branch is $\sat$, we execute the \myinline{then_} thunk (line~\ref{line:exec-left}).
Note that the entire execution of the left branch is wrapped in a call to \myinline{L.with_section}, which ensures that all log messages generated during the execution of the then branch are grouped together in a collapsible section in the log file, making it easier to navigate individual branches in the logs, as presented in \cref{app:logging-impl}.

After executing the then branch (when the guard is $\sat$), we backtrack to the saved state and analogously execute the else branch, with two key differences.
First, if the then branch was $\unsat$, we are guaranteed that the else branch is $\sat$ ($\unsat(\pc) \implies \sat (\neg \pc)$, for all $\pc$), and hence do not need to call the solver again, thus saving a potentially expensive call to the solver.
Second, if the then branch was executed, we consume one unit of branching fuel and stop if we are out of fuel.
\sot comes with a simple built-in fuel mechanism to limit the number of branches explored and steps performed during symbolic execution (where default fuel is infinite).
If the branch is not executed due to lack of fuel, we record this in the execution statistics.

\subsection{Logging in \sot}
\label{app:logging-impl}

\sot is designed to come with \emph{batteries included}, and provides out-of-the-box support for logging and statistics.
In fact, logs in \sot are often more informative, as they relate directly to the relevant source code, rather than to the intermediate language.
Logging is directly integrated in the symbolic execution: users can log messages at any point of their programs, and \sot can create an HTML report that groups logs by execution path.

For instance, adding the following logging statement to our symbolic interpreter (above) produces the HTML report shown below in \cref{fig:soteria-logging}.

\begin{figure}[!h]
\vspace{-8pt}
\begin{lstlisting}
let rec eval subst expr =
  L.verbose (fun m -> m
    "Evaluating expression: %a"
    Expr.pp expr);
  match expr with
  | ... -> ...
\end{lstlisting}%
\includegraphics[width=0.63\textwidth]{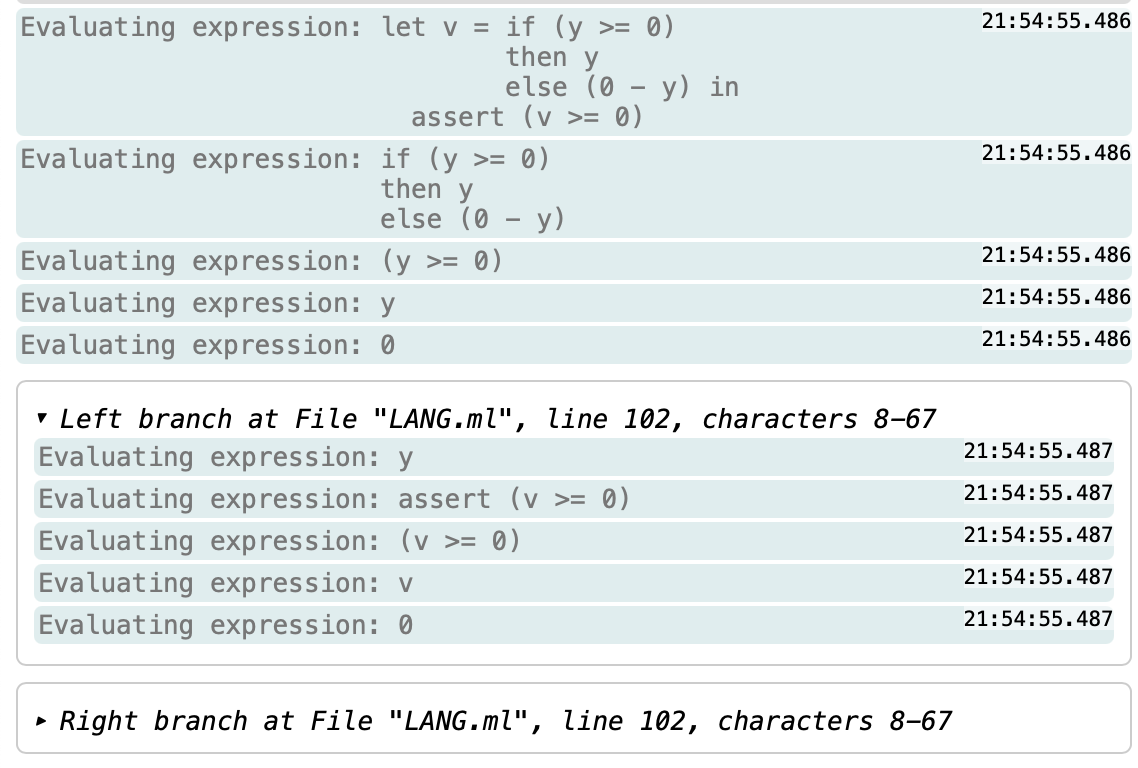}%
\vspace{-10pt}
\caption{Example of logging in \sot}\label{fig:soteria-logging}
\end{figure}

The logging API is provided out of the box by \sot, and provides multiple logging levels (error, warning, info, debug, trace, SMT).
\section{A Simple BTreeMap Example for \rusteria}
\label{app:btreemap}

In this appendix, we present the details a simple canonical example of a bounded proof, which inserts $n$ nondeterministic elements into a BTreeMap and checks that it is well-sorted. We compare the performance of Soteria and Kani on this example for various values of $n$, noting that Soteria performs more checks than Kani, e.g. uninit. memory and TreeBorrows.

\begin{table}[h]
  \centering
  \begin{tabular}{rrrr}
    \toprule
    $n$ & \rusteria time & branches & Kani time \\
    \midrule
    1 & 1.55s  &    1 & 2.70s    \\
    2 & 1.62s  &    3 & 9.09s    \\
    3 & 2.09s  &   13 & 891.65s  \\
    4 & 6.09s  &   75 & $>$2h    \\
    5 & 43.91s &  541 & $>$2h    \\
    6 & 456.30s & 4683 & $>$2h   \\
    \bottomrule
  \end{tabular}
  \caption{Performance of \rusteria and Kani on the \myinline{BTreeSet} benchmark
    (inserting $n$ nondeterministic integers and checking sortedness).
    \rusteria additionally checks uninitialised memory and Tree Borrows violations.}
  \label{tab:btreemap}
\end{table}

\begin{figure}
\begin{lstlisting}[language=Rust,style=ruststyle]
// btreeset.rs
use std::collections::BTreeSet;
const BOUND: usize = 2;

#[kani::proof]
#[kani::unwind(3)]
fn check_sorted() {
    let mut map: BTreeSet<i32> = BTreeSet::new();
    for _ in 0..BOUND {
        let k: i32 = kani::any();
        map.insert(k);
    }
    let mut prev: Option<i32> = None;
    for &k in map.iter() {
        if let Some(p) = prev {
            assert!(p < k);
        }
        prev = Some(k);
    }
}
\end{lstlisting}
\caption{The \myinline{test_btreeset.rs} file.}\label{lst:btreemap-test}
\end{figure}

The complete test file is given in \cref{lst:btreemap-test}. We insert \myinline{BOUND} nondeterministic elements into a \myinline{BTreeSet} (an ordered binary tree), which causes the engine to explore all potential orderings of the nondeterministic values.
This quickly causes an explosion in the number of cases considered, as demonstrated by the branch counts in \Cref{tab:btreemap}.

\begin{itemize}
  \item The variable $n$ in \Cref{tab:btreemap} corresponds to the \myinline{BOUND} variable in the code above.
  \item For each \myinline{BOUND}, we set the \myinline{kani::unwind} bound to \myinline{BOUND + 1}, the minimum unwinding required for the proof to pass, so that Kani does no additional work. This annotation is ignored by \rusteria, which always performs the minimal work.
  \item Both Kani and \rusteria timings include compilation from Rust to MIR, which takes 1.2--1.5\,s in \rusteria (explaining why $n=1$ and $n=2$ are so close; the respective analysis times are 0.1\,s and 0.17\,s).
  \item This test uses the \myinline{BTreeSet} implementation from the Rust standard library.
\end{itemize}

We also provide additional statistics for the \myinline{BOUND=6} case, produced by the \myinline{--dump-stats} flag of \rusteria via the \myinline{Soteria.Stats} module:
\begin{lstlisting}
{
"solvers.z3.check_sats":        20867,  // sat checks reaching Z3
"soteria.branch-on-calls":  313870067,  // calls to if%sat
"soteria.branches":              4683,  // feasible branches taken
"soteria.exec-time":          472.26s,  // total execution time (excl. parsing)
"soteria.data.pmap.lookups": 19650835,  // calls to Data.Map.find_opt
"soteria.steps":              2525904,  // MIR blocks executed
"soteria-rust.function_calls": 604368,  // function calls executed
"soteria.sat-checks":           22358,  // sat checks (incl. pre-Z3 solving)
"soteria.sat-time":           110.12s,  // total time spent in Z3
}
\end{lstlisting}

For reference, we also give part of the output of Kani for \myinline{BOUND=3}, which illustrates the complexity of the SAT instance it generates:
\begin{lstlisting}
Runtime Symex: 12.7424s
size of program expression: 532625 steps
slicing removed 369261 assignments
Generated 17577 VCC(s), 9365 remaining after simplification
Runtime Postprocess Equation: 0.20912s
converting SSA:    Runtime Convert SSA: 10.1436s
Runtime Post-process: 27.0632s
Solving with CaDiCaL 2.0.0
50374973 variables, 254600597 clauses
SAT checker: SATISFIABLE  Runtime Solver: 113.015s
Solving with CaDiCaL 2.0.0
50374974 variables, 254600598 clauses
SAT checker: SATISFIABLE  Runtime Solver: 102.966s
Solving with CaDiCaL 2.0.0
50374975 variables, 254600599 clauses
SAT checker: UNSATISFIABLE  Runtime Solver: 691.116s
\end{lstlisting}
\clearpage
\section{\rusteria: Supplementary Implementation and Evaluation Details}
\label{app:rusteria}
\subsection{\rusteria Value language}\label{sec:rusteria-values}

The value language for \rusteria is defined in \cref{fig:rust-values}. It is parameterised by a type \myinline{'ptr} representing pointers; to allow multiple state model implementations. The currently definition of \myinline{'ptr} in \rusteria is presented in \Cref{subsec:rusteria_engine}, and is included below for convenience.

\begin{figure}[h]
\begin{lstlisting}[escapeinside={&}{&}]
type 'ptr meta = Thin | Len of BV_value.(int t) | VTable of 'ptr

type 'ptr t =
  | Base of BV_value.(val t)  | Ptr of 'ptr * 'ptr meta
  | Tuple of 'ptr t list      | Union of Types.field_id * 'ptr t
  | Enum of BV_value.(val t) * 'ptr t list
  | ConstFn of Types.fn_ptr

type ptr_&\sot{\textsuperscript{\textsc{rust}}}& = {
  ptr : BV_value.(ptr t);
  tag : Tree_borrow.tag;
  align : BV_value.(nonzero t);
  size : BV_value.(int t);
}
\end{lstlisting}
\caption{Rust values}\label{fig:rust-values}
\end{figure}

In the \myinline{Union} case, \myinline{Types.field_id} refers to what field of the union's definition the value refers to, which is needed to not lose typing information --- this is an implementation detail, and future work will aim to remove it and replace it with a list of possibly uninitialised bytes with provenance, akin to MiniRust~\cite{minirust-jung-2025}.
In the \myinline{ConstFn} case, \myinline{Types.fn_ptr} is a function identifier (and not a pointer, despite the name) used to look up the function definition in the program being symbolically executed.

Our value language is similar to that of MiniRust \cite{minirust-jung-2025}, as an attempt to model Rust's semantics faithfully, the exceptions being unions not represented faithfully, and the lack of floats or function references in MiniRust.

\subsection{AI Standard Library Test Prompt}\label{app:rusteria-tests-prompt}

To remove human bias from the test writing process, we used an AI pair programmer to write the tests for \rusteria.
We used the same prompt for each data structure tested.
The prompt we used is as follows:

{\itshape
Using the Kani library, write symbolic tests for the \myinline{std::collections::DataStructure} data structure in Rust.
The library is provided.
You can use \myinline{kani::any()} to generate non-deterministic variables of a given type and \myinline{kani::assume()} to assume properties in the symbolic execution.
Use \myinline{assert!()} to check for the desired properties.

Separate each test in a different function, with the \myinline{#[kani::proof]} attribute.
You may also define failing tests, by adding \myinline{#[kani::should_panic]}, to test behaviour of the \myinline{DataStructure} that would panic.
Define a function \myinline{any_datastructure<T, SIZE>() -> DataStructure<T>} to generate a non-deterministic \myinline{DataStructure} that you use in the tests.
Do not tailor the test to Kani's bounded nature, just write the tests in an intuitive way.
}
\clearpage
\section{Formalisation}\label{app:formalisation}


\paragraph[0pt]{Symbolic Interpretation Extension}
Let $\sint_1, \sint_2$ be symbolic interpretations; $\sint_2$ is an \emph{extension} of $\sint_1$, written $\sint_2 \geq \sint_1$, if $\for{\sx, v} \sint_1(\sx) = \val \Rightarrow \sint_2(\sx) = \val$.

\subsection{Monad Laws}

We show that our symbolic execution monad satisfies the three monad laws.

\begin{definition}[Symbolic Execution Monad]\label{def:symex-monad}
We define the symbolic execution monad as the triple $(\symexM, \mathtt{return}, \mathtt{bind})$ as follows:

\[
\begin{array}{rcl}
	\symexM(A) & \eqdef & \pcs \rightarrow \poset(\braket{\symb{A}}{\pcs})  \\
    \mathtt{return}(a) & \eqdef & \lambda \pc.~\{ \braket{a}{\pc} \} \\
	\mathtt{bind}(m,f) & \eqdef & \lambda \pc.~ \{ \braket{b}{\pc''} \mid \braket{a}{\pc'} \in m(\pc) \land \braket{b}{\pc'} \in f(a, \pc') \}
\end{array}
\]
\end{definition}

\begin{lemma}[Symbolic execution monad: Monad Laws]\label{lem:symex-monad-laws}
The symbolic execution induced by the bind operator defined in \Cref{def:symex-monad} satisfies the monad laws.
\end{lemma}
\begin{proof}
This follows from the fact $\symexM$ is a state monad transformer applied to the set monad, so it trivially follows from monad transformers that the monad laws hold.
\end{proof}

\subsection{Soundness}

%

\begin{definition}[Symbolic branch]\label{def:symbolic-branch}

A symbolic branch is a pair composed of a symbolic abstraction and a path condition and is denoted $\braket{\symb{a}}{\pc} \in \braket{\symb{A}}{\pcs}$. Symbolic branches themselves are symbolic abstractions over $A$, where

\[
\sint, a \models \braket{\symb{a}}{\pc} \iff \sint, a \models \symb{a} \land \sint, \vtrue \models \pc
\]

\end{definition}

\begin{definition}[$\mox$ soundness]
Let $f : A \rightarrow \poset(B)$ be a nondeterministic function, $\symb{A}$ and $\symb{B}$ be symbolic abstrations over $A$ and $B$, and $\symb{f} : \symb{A} \rightarrow \symexM(\symb{B})$. $\symb{f}$ is $\mox$-sound \wrt $f$, denoted $\symb{f} \soundwrtox f$, if each outcome of $f$ is covered by a branch of $\symb{f}$ (is also an outcome of $\symb{f}$):
\[
f(a) \rightsquigarrow b \land \sint, a \models \braket{\symb{a}}{\pc} \implies
 \exists \symb{b}, \pc', \sint' \geq \sint.~ \symb{f}(\symb{a}, \pc) \rightsquigarrow \braket{\symb{b}}{\pc'} \land \sint', b \models \braket{\symb{b}}{\pc'}
\]
where $f(a) \rightsquigarrow b$ means that $b \in f(a)$.
\end{definition}

Note that the interpretation $\sint'$ is an extension of $\sint$, in order to take into account that $\symb{f}$ may introduce fresh variables that are not covered by $\sint$.

\begin{definition}[$\mux$ soundness]
$\symb{f}$ is $\mux$-sound \wrt $f$, denoted $\symb{f} \soundwrtux f$ if:
\[
\begin{array}{l}
\symb{f}(\symb{a}, \pc) \rightsquigarrow \braket{\symb{b}}{\pc'} \implies \\
\quad \left( \sat(\pc') \land \forall \sint.~ \sint, \vtrue \models \pc' \Rightarrow \exists b.~ \sint, b \models \symb{b} \right. \\
\qquad \left. \land (\forall b.~ \sint, b \models \braket{\symb{b}}{\pc'} \Rightarrow (\exists a.~ \sint, a \models \braket{\symb{a}}{\pc} \land f(a) \rightsquigarrow b))   \right)
\end{array}
\]
\end{definition}

Let us walk through this more complex definition. It states that, if the UX-sound symbolic computation has a transition $\symb{f}(\symb{a}, \pc) \rightsquigarrow \braket{\symb{b}}{\pc'}$, then: \begin{itemize}
\item the resulting path condition $\pc'$ must be satisfiable, \ie at some point it must be checked by a ($\mux$-approximate) solver;
\item if $\sint$ is a set of bindings that satisfy $\pc'$, then there must exist at least one model $b$ of the result $\symb{b'}$. This ensures that checking the path condition for satisfiability is sufficient to ensure satisfiability of the branch; and
\item for all models $b$ of $\braket{\symb{b}}{\pc}$, there is a corresponding model $a$ of $\braket{\symb{a}}{{\pc}}$ that form an input-output pair of $f$.
\end{itemize}

An important property always upheld by $\mux$-sound symbolic computations is that they are fully characterised by their path conditions: all invalid paths can be pruned by looking only at the path condition, which itself can be encoded into a solver. We believe this required for an efficient implementation of a bug-finder.

A simple example of a both $\mox$- and $\mux$-sound symbolic computation is the \lstinline|nondet|, which is sound with respect to the nondeterministic function which returns all values of the given sort. More complex symbolic computations are built by composing several such simple symbolic computations using $\mathtt{bind}$ and $\mathtt{branch\_on}$. We show that these two operations \emph{preserve} soundness.

\begin{theorem}[Soundness Preservation: Composition]\label{thm:bind-sound-app}
For a given mode $\mode \in \{ \mox, \mux \}$
\[
\symb{f} \soundwrtm f \land \symb{g} \soundwrtm g \implies \symb{g} \fishop \symb{f} \soundwrtm g \fishop f
\]
Where $p \fishop q = \lambda x.~ \mathtt{bind}~p(x)~q$ is the Kleisli composition operator (or ``fish'' operator)
\end{theorem}

\Cref{thm:bind-sound-app} states that composing two $m$-sound symbolic computations using $\mathtt{bind}$ yields a symbolic computation.

\begin{proof}
\pfcase{OX soundness preservation}
\begin{hypvlist}
\hypvitem{1} $f(a) \rightsquigarrow b \land \sint, a \models \braket{\symb{a}}{\pc}$

$\qquad \implies  \exists \symb{b}, \pc', \sint' \geq \sint.~ \symb{f}(\symb{a},\pc) \rightsquigarrow \braket{\symb{b}}{\pc'} \land \sint', b \models \braket{\symb{b}}{\pc'}$

\hypvitem{2} $g(b) \rightsquigarrow c \land \sint', b \models \braket{\symb{b}}{\pc'}$

$\qquad \implies  \exists \symb{c}, \pc'', \sint'' \geq \sint'.~ \symb{g}(\symb{b}, \pc') \rightsquigarrow \braket{\symb{c}}{\pc''} \land \sint'', c \models \braket{\symb{c}}{\pc''}$

\hypvitem{3} $\inferrule{
f(a) \rightsquigarrow b\\\\
g(b) \rightsquigarrow c
}{
h(a) \rightsquigarrow c
}$
\qquad
\newhyp{4} $\inferrule{
\symb{f}(\symb{a},\pc) \rightsquigarrow \braket{\symb{b}}{\pc'}\\\\
\symb{g}(\symb{b},\pc') \rightsquigarrow \braket{\symb{c}}{\pc''}
}{
\symb{h}(\symb{a},\pc) \rightsquigarrow \braket{\symb{c}}{\pc''}}$
\end{hypvlist}

\pflet $a \in A, c \in C, \sint \in \sints$ such that
\begin{hyphlist}
\newhyp{5}\quad h(a) \rightsquigarrow c
\and
\newhyp{6}\quad \sint, a \models \symb{a}
\end{hyphlist}

\pfprove \newgoal{1} $\exists \symb{c}, \pc,\pc'', \sint'' \geq \sint.~ \symb{h}(\symb{a},\pc) \rightsquigarrow \braket{\symb{c}}{\pc''} \land \sint'', c \models \braket{\symb{c}}{\pc''}$

From \hyp{3} and \hyp{5} we get that $\exists b.$
\begin{hyphlist}
\newhyp{8}\quad f(a) \rightsquigarrow b
\and
\newhyp{9}\quad g(b) \rightsquigarrow c
\end{hyphlist}

From \hyp{4} we have
\begin{hyphlist}
\newhyp{10.1}\quad \forall \symb{b}. \braket{\symb{b}}{\pc'} \in \symb{f}(\symb{a},\pc) \implies \symb{g}(\symb{b},\pc') \rightarrow -
\end{hyphlist}

From \hyp{1}, \hyp{6} and \hyp{8}, $\exists \symb{b}, \pc', \sint'.$
\begin{hyphlist}
\newhyp{sintp}\quad \sint' \geq \sint
\and
\newhyp{11}\quad \symb{f}(\symb{a},\pc) \rightsquigarrow \braket{\symb{b}}{\pc'}
\and
\newhyp{12}\quad \sint, b \models \braket{\symb{b}}{\pc'}
\end{hyphlist}

Then from \hyp{2}, \hyp{9}, \hyp{12} and \hyp{10.1}, $\exists \symb{c}, \pc'', \sint''.$
\begin{hyphlist}
\newhyp{sintpp}\quad \sint'' \geq \sint'
\and
\newhyp{13}\quad \symb{g}(\symb{b},\pc') \rightsquigarrow \braket{\symb{c}}{\pc''}
\and
\newhyp{14}\quad \sint, c \models \braket{\symb{c}}{\pc''}
\end{hyphlist}

Moreover, from \hyp{sintp}, \hyp{sintpp}, \hyp{12} and \hyp{14}, we have that $\pc(\sint'') = \vtrue$ and $\pc''(\sint'') = \vtrue$, and therefore:
\[
\newhyp{15}\quad (\pc \land \pc'')(\sint'') \eqdef \pc(\sint'') \land \pc''(\sint'') = \vtrue
\]

Additionally, from \hyp{15} we get
\[
\newhyp{15'}\quad \sat_{\mox}(\pc \land \pc'')
\]

From \hyp{4}, \hyp{11}, \hyp{13} and \hyp{15'} we obtain
\[
\newhyp{16}\quad \symb{h}(\symb{a},\pc) \rightsquigarrow \braket{\symb{c}}{\pc \land \pc''}
\]

Finally, from \hyp{14}, \hyp{15} and \hyp{16} we obtain \goal{1}\qed

\pfcase{UX soundness preservation}\resetproofcounters

\begin{hypvlist}

\hypvitem{1} $\symb{f}(\symb{a},\pc) \rightsquigarrow \braket{\symb{b}}{\pc'} \implies (\sat(\pc') ~\land$

$\qquad \ \ \forall \sint.~ \sint,\vtrue\models\pc' \Rightarrow (\exists b.~ \sint, b \models \symb{b} \land$

$\qquad \quad \quad \forall b.~ \sint,b \models \braket{\symb{b}}{\pc'} \Rightarrow (\exists a.~ \sint, a \models \braket{\symb{a}}{\pc} \land f(a)\rightsquigarrow b )))$

\hypvitem{2} $\symb{g}(\symb{b},\pc') \rightsquigarrow \braket{\symb{c}}{\pc''} \implies (\sat(\pc'') ~\land$

$\qquad \ \ \forall \sint'.~ \pc''(\sint') = \vtrue \Rightarrow (\exists c.~ \sint', c \models \symb{c} \land$

$\qquad \quad \quad \forall c.~ \sint',c \models \braket{\symb{c}}{\pc''} \Rightarrow (\exists b.~ \sint', b \models \braket{\symb{b}}{\pc'} \land g(b)\rightsquigarrow c )))$

\hypvitem{3} $\inferrule{
f(a) \rightsquigarrow b\\\\
g(b) \rightsquigarrow c
}{
h(a) \rightsquigarrow c
}$
\qquad
\newhyp{4} $\inferrule{
\symb{f}(\symb{a},\pc) \rightsquigarrow \braket{\symb{b}}{\pc'}\\\\
\symb{g}(\symb{b},\pc') \rightsquigarrow \braket{\symb{c}}{\pc''}
}{
\symb{h}(\symb{a}) \rightsquigarrow \braket{\symb{c}}{\pc \land \pc' \land \pc''}
}$
\end{hypvlist}

\pflet $\symb{a}, \symb{c}, \pc, \pc'''$ such that
\begin{hypvlist}
\hypvitem{5} $\symb{h}(\symb{a},\pc) \rightsquigarrow \braket{\symb{c}}{\pc'''}$
\end{hypvlist}

\pfprove
\begin{goalvlist}
\goalvitem{1} $
\forall \sint.~ \pc'''(\sint) = \vtrue \Rightarrow (\exists c.~ \sint, c \models \symb{c} \land$

$\qquad \qquad\forall c.~ \sint,c \models \symb{c} \Rightarrow (\exists a.~ \sint, a \models \braket{\symb{a}}{\pc} \land h(a)\rightsquigarrow b ))$

\goalvitem{2}\quad $\sat(\pc''')$
\end{goalvlist}

Before focusing on either goal, we start by establishing a few facts we can learn from \hyp{4} and \hyp{5}: $\exists \symb{b}, \pc, \pc', \pc''$
\begin{hyphlist}
\newhyp{7}\quad \symb{f}(\symb{a},\pc) \rightsquigarrow \braket{\symb{b}}{\pc'}
\and
\newhyp{8}\quad \symb{g}(\symb{b},\pc') \rightsquigarrow \braket{\symb{c}}{\pc''}
\and
\newhyp{9}\quad \pc''' = \pc \land \pc' \land \pc''
\and
\newhyp{10}\quad \sat_{\mux} (\pc \land \pc' \land \pc'')
\end{hyphlist}

From \hyp{10} and the definition of approximate UX solvers, we have\begin{hyphlist}
  \newhyp{6'}~\sat (\pc'')\and
  \newhyp{6''}~\sat(\pc')\and
  \newhyp{6'''}~\sat(\pc)
\end{hyphlist}

\hyp{9} and \hyp{10} immediately gives \goal{2}. Let us focus on proving \goal{1}. \pflet $\sint$ such that $\newhyp{6}~ \pc'''(\sint) = \vtrue$. Our new goals are: \begin{goalvlist}
\goalvitem{3} $\exists c.~ \sint, c \models \symb{c}$
\goalvitem{4} $\forall c.~ \sint,c \models \symb{c} \Rightarrow (\exists a.~ \sint, a \models \braket{\symb{a}}{\pc} \land h(a)\rightsquigarrow b )$
\end{goalvlist}

Goal \goal{3} is immediate from \hyp{2}, \hyp{6'} and \hyp{8}, so we focus on \goal{4}.

\pflet $c$ such that $\newhyp{11}~ \sint, c \models \symb{c}$.

From \hyp{2}, \hyp{6''} and \hyp{11}, we have $\exists b.$
\begin{hyphlist}
\newhyp{12}~ \sint, b \models \braket{\symb{b}}{\pc'}
\and
\newhyp{13}~ g(b) \rightsquigarrow c
\end{hyphlist}

From \hyp{1}, \hyp{7}, \hyp{6'''} and \hyp{12} we also have $\exists a.$
\begin{hyphlist}
\newhyp{14}~ \sint, a \models \braket{\symb{a}}{\pc}
\and
\newhyp{15}~ f(a) \rightsquigarrow b
\end{hyphlist}

Together, \hyp{13} and \hyp{15} entail \newhyp{16} $h(a) \rightsquigarrow c$, which, together with \hyp{14} forms \goal{4}.
\end{proof}

\vspace{10pt} 

\begin{theorem}[Soundness Preservation: Branching]
\label{thm:branch-sound-app}
Let $h$ and $\symb{h}$ be defined as \[
\begin{array}{lll}
h~(b, a) = \mathtt{if}~b~\mathtt{then}~f~a~\mathtt{else}~g~a
&\quad&
\symb{h}~(\symb{b}, \symb{a}) = \mathtt{branch\_on}~\symb{b}~(\symb{f}~\symb{a})~(\symb{g}~\symb{a})
\end{array}
\]
If $\mathtt{branch\_on}$ uses a $\mode$-approximate solver $\sat_\mode$, then
\[\symb{f} \soundwrtm f \land \symb{g} \soundwrtm g \implies \symb{h} \soundwrt h\]
where $\sint, (b, a) \models (\symb{b}, \symb{a}) \iff \sint, b \models \symb{b} \Leftrightarrow \sint, a \models \symb{a}$
\end{theorem}

\Cref{thm:branch-sound-app} states that $\mathtt{branch\_on}$ is effectively a symbolic lifting of the concrete \lstinline|if...else...| construct.

\begin{proof}
\pfcase{OX soundness preservation}

\begin{hypvlist}

\hypvitem{t-app} $\forall a, c, \pc, \symb{a}, \sint.~ f(a) \rightsquigarrow c \land \sint, a \models \braket{\symb{a}}{\pc}$

$\qquad \implies  \exists \symb{c}, \pc', \sint' \geq \sint.~ \symb{f}(\symb{a}, \pc) \rightsquigarrow \braket{\symb{c}}{\pc'} \land \sint', c \models \braket{\symb{c}}{\pc'}$

\hypvitem{e-app} $\forall a, c, \pc, \symb{a}, \sint.~ g(a) \rightsquigarrow c \land \sint, a \models \braket{\symb{a}}{\pc}$

$\qquad \implies  \exists \symb{c}, \pc', \sint' \geq \sint.~ \symb{g}(\symb{a}, \pc) \rightsquigarrow \braket{\symb{c}}{\pc'} \land \sint', c \models \braket{\symb{c}}{\pc'}$

\hypvitem{defconc} $\inferrule{
f(a) \rightsquigarrow c
}{
h(\vtrue, a) \rightsquigarrow c
}$
\qquad
$\inferrule{
g(a) \rightsquigarrow c
}{
h(\vfalse, a) \rightsquigarrow c
}$

\hypvitem{defsym} $\inferrule{
\symb{f}(\symb{a}, \pc) \rightsquigarrow \braket{\symb{c}}{\pc'}\\\\
\sat_{\mox} (\pc_b \land \pc')
}{
\symb{h}(\pc_b, \symb{a},\pc) \rightsquigarrow \braket{\symb{c}}{\pc_b \land \pc'}
}$
\qquad
$\inferrule{
\symb{g}(\symb{a},\pc) \rightsquigarrow \braket{\symb{c}}{\pc'}\\\\
\sat_{\mox} (\neg\pc_b \land \pc')
}{
\symb{h}(\pc_b, \symb{a}, \pc) \rightsquigarrow \braket{\symb{c}}{\neg\pc_b \land \pc'}
}$
\end{hypvlist}

\pflet $a, b, c, \symb{a}, \sint$ such that:
\begin{hyphlist}
\newhyp{hexec}~ h(b, a) \rightsquigarrow c
\and
\newhyp{amodel}~ \sint, a \models \symb{a}
\and%
\newhyp{bmodel}~ \pc_b(\sint) = b%
\end{hyphlist}

\pfprove \newgoal{goal} $\exists \symb{c}, \pc_f, \sint' \geq \sint.~ \symb{h}(\pc_b, \symb{a}, \pc) \rightsquigarrow \braket{\symb{c}}{\pc_f} \land \sint, c \models \braket{\symb{c}}{\pc_f}$

There are two cases to consider: either $b$ is $\vtrue$ or it is $\vfalse$. We only consider the $\vtrue$ case, the $\vfalse$ one being analogous.

\pfassume \newhyp{btrue} $b = \vtrue$

From \hyp{defconc}, \hyp{hexec} and \hyp{btrue}, we have that necessarily: \[\newhyp{ftc}~f(a) \rightsquigarrow c\]

Using \hyp{btrue}, \hyp{amodel} and \hyp{bmodel} we can apply \hyp{t-app} and obtain $\exists \pc, \pc', \symb{c}, \sint'$
\begin{hyphlist}
\newhyp{sintp}~ \sint' \geq \sint
\and
\newhyp{sfres}~ \symb{f}(\symb{a},\pc) \rightsquigarrow \braket{\symb{c}}{\pc'}
\and
\newhyp{cmodc}~ \sint', c \models \braket{\symb{c}}{\pc'}
\end{hyphlist}

From \hyp{cmodc} and \Cref{def:symbolic-branch}, we know that $\pc'(\sint') = \vtrue$, which combined with \hyp{sintp}, \hyp{bmodel} and \hyp{btrue} gives us
\newhyp{alltrue} $(\pc_b \land \pc')(\sint') = \vtrue$, which naturally implies $\sat(\pc_b \land \pc')$, and therefore \newhyp{allsat} $\sat_{\mox}(\pc_b \land \pc')$.

Additionally, from \hyp{cmodc}, \hyp{alltrue} and the definition of branch satisfiability, we also have that \newhyp{cmodcfinal} $\sint', c \models \braket{\symb{c}}{\pc_b \land \pc'}$

From \hyp{allsat} and \hyp{sfres} we can apply \hyp{defsym} and obtain \newhyp{shres} $\symb{h}(\pc_b, \symb{a}, \pc) \rightsquigarrow \braket{\symb{c}}{\pc_b \land \pc'}$

Together, \hyp{cmodcfinal} and \hyp{shres} form our goal \goal{goal}.

\pfcase{UX soundness preservation}\resetproofcounters

\begin{hypvlist}

\hypvitem{ft-app} $\forall \symb{a}, \symb{c}, \pc, \pc'.~ \symb{f}(\symb{a}, \pc) \rightsquigarrow \braket{\symb{c}}{\pc'} \implies (\sat(\pc') ~\land$

$\qquad \quad \ \ \forall \sint.~ \sint,\vtrue\models\pc' \Rightarrow (\exists c.~ \sint, c \models \symb{c} \land$

$\qquad \quad \quad \forall c.~ \sint,c \models \braket{\symb{c}}{\pc'} \Rightarrow (\exists a.~ \sint, a \models \braket{\symb{a}}{\pc} \land f(a)\rightsquigarrow c )))$

\hypvitem{fe-app} $\forall \symb{a}, \symb{c}, \pc, \pc'.~ \symb{g}(\symb{a},\pc) \rightsquigarrow \braket{\symb{c}}{\pc'} \implies (\sat(\pc') ~\land$

$\qquad \quad \ \ \forall \sint.~ \sint,\vtrue\models\pc' \Rightarrow (\exists c.~ \sint, c \models \symb{c} \land$

$\qquad \quad \quad \forall c.~ \sint,c \models \braket{\symb{c}}{\pc'} \Rightarrow (\exists a.~ \sint, a \models \braket{\symb{a}}{\pc} \land g(a)\rightsquigarrow c )))$

\hypvitem{defconc} $\inferrule{
f(a) \rightsquigarrow c
}{
h(\vtrue, a) \rightsquigarrow c
}$
\qquad
$\inferrule{
g(a) \rightsquigarrow (c)
}{
h(\vfalse, a) \rightsquigarrow c
}$

\hypvitem{defsym} $\inferrule{
\symb{f}(\symb{a},\pc) \rightsquigarrow \braket{\symb{c}}{\pc'}\\\\
\sat_{\mux} (\pc_b \land \pc')
}{
\symb{h}(\pc_b, \symb{a},\pc) \rightsquigarrow \braket{\symb{c}}{\pc_b \land \pc'}
}$
\qquad
$\inferrule{
\symb{g}(\symb{a},\pc) \rightsquigarrow \braket{\symb{c}}{\pc'}\\\\
\sat_{\mux} (\neg\pc_b \land \pc')
}{
\symb{h}(\pc_b, \symb{a},\pc) \rightsquigarrow \braket{\symb{c}}{\neg\pc_b \land \pc'}
}$
\end{hypvlist}

\pflet $\symb{a}, \pc_b, \pc, \symb{c}, \pc_f$ such that
\begin{hypvlist}
\hypvitem{symhres} $\symb{h}(\pc_b, \symb{a}, \pc) \rightsquigarrow \braket{\symb{c}}{\pc_f}$
\end{hypvlist}

\pfprove
\begin{goalvlist}
\goalvitem{reach} $\sat(\pc_f)$
\goalvitem{ux} $\forall \sint.~ \sint,\vtrue\models\pc_f \Rightarrow (\exists c.~ \sint, c \models \symb{c} \land$

$\qquad \quad \quad \forall c.~ \sint,c \models \braket{\symb{c}}{\pc'} \Rightarrow (\exists a,b .~ \sint, a \models \symb{a} \land \pc_b(\sint) = b \land h(b, a)\rightsquigarrow c ))$
\end{goalvlist}

From \hyp{symhres}, and by inversions on the rules given in \hyp{defsym}, there are two ways the result could have been obtained. We consider the case where the ``then'' branch was taken, the other being analogous.

\pfcase{then branch taken}
\begin{hyphlist}
\newhyp{symf}~\symb{f}(\symb{a},\pc) \rightsquigarrow \braket{\symb{c}}{\pc'}
\and
\newhyp{muxsat}~\sat_{\mux}(\pc_b \land \pc')
\end{hyphlist}

The goal \goal{reach} is immediately achieved by \hyp{muxsat} and the definition of approximate UX solvers. In addition, we also get:
\begin{hyphlist}
\newhyp{pcpsat}~\sat(\pc')
\and
\newhyp{pcsat}~\sat(\pc_b)
\end{hyphlist}

We now focus on \goal{ux}. \pflet $\sint$ such that \newhyp{spcftrue} $\pc_f(\sint) = \vtrue$. This implies that \begin{hyphlist}
  \newhyp{spctrue}~ \pc_b(\sint) = \vtrue
  \and
  \newhyp{spcptrue}~ \pc'(\sint) = \vtrue
\end{hyphlist}

From \hyp{ft-app}, \hyp{symf} and \hyp{spcptrue}, we have $\exists c.$
\begin{hyphlist}
\newhyp{cmod}~\sint, c \models \braket{\symb{c}}{\pc'}
\end{hyphlist}
We take such a $c$. From \hyp{ft-app}, \hyp{symf}, \hyp{spcptrue} and \hyp{cmod}, we also have $\exists a.$
\begin{hyphlist}
\newhyp{amod}~\sint, a \models \symb{a}
\and
\newhyp{fexec}~f(a) \rightsquigarrow c
\end{hyphlist}

from \hyp{fexec} and \hyp{spctrue}, we get that \[
\newhyp{ch}~ h(b, a) \rightsquigarrow c
\]

which concludes the proof of \goal{ux}.

\pfcase{else branch taken} Analogous to the previous case.

\end{proof}
}{}

\end{document}